\documentclass[11pt]{article}
\usepackage[abs]{overpic}
\usepackage[utf8]{inputenc}
\usepackage{a4wide}
\usepackage{graphicx}
\usepackage[hmarginratio={1:1},     
  vmarginratio={1:1},     
  textwidth=16.5cm,        
  textheight=24cm,
  heightrounded,]{geometry}

\usepackage{hhline}
  
\usepackage{amsmath,accents}
\usepackage{amsthm}
\usepackage{amssymb}
\usepackage{stmaryrd}
\usepackage{tikzit}

\tikzstyle{black dot}=[fill=black, draw=black, shape=circle, minimum size=3pt, inner sep=0pt]
\tikzstyle{black dot small}=[fill=black, draw=black, shape=circle, minimum size=3pt, inner sep=0pt]
\tikzstyle{big white circle}=[fill=white, draw=black, shape=circle, minimum width=0.75cm]
\tikzstyle{white dot big}=[fill=white, draw=black, shape=circle, inner sep=1pt]
\tikzstyle{white dot}=[fill=white, draw=black, shape=circle, minimum size=3pt, inner sep=0pt]
\tikzstyle{flat box}=[fill=white, draw=black, shape=rectangle, minimum width=2.5cm, minimum height=0.5cm]
\tikzstyle{square}=[fill=white, draw=black, shape=rectangle]
\tikzstyle{flat box 2}=[fill=white, draw=black, shape=rectangle, minimum height=0.5cm, minimum width=1.0cm]

\tikzstyle{mid arrow}=[-, postaction={on each segment={mid arrow}}]
\tikzstyle{end arrow}=[->, >=latex]
\tikzstyle{red mid arrow}=[-, draw={rgb,255: red,214; green,42; blue,51}, postaction={on each segment={mid arrow}}, line width=1pt]
\tikzstyle{blue}=[-, draw={rgb,255: red,23; green,37; blue,167}, line width=1pt, dashed]
\tikzstyle{blue mid arrow}=[-, draw={rgb,255: red,23; green,37; blue,167}, postaction={on each segment={mid arrow}}, line width=1pt]

\usetikzlibrary{decorations.pathreplacing,decorations.markings}
\tikzset{
  on each segment/.style={
    decorate,
    decoration={
      show path construction,
      moveto code={},
      lineto code={
        \path [#1]
        (\tikzinputsegmentfirst) -- (\tikzinputsegmentlast);
      },
      curveto code={
        \path [#1] (\tikzinputsegmentfirst)
        .. controls
        (\tikzinputsegmentsupporta) and (\tikzinputsegmentsupportb)
        ..
        (\tikzinputsegmentlast);
      },
      closepath code={
        \path [#1]
        (\tikzinputsegmentfirst) -- (\tikzinputsegmentlast);
      },
    },
  },
  mid arrow/.style={postaction={decorate,decoration={
        markings,
        mark=at position .5 with {\arrow[#1]{stealth}}
      }}},
}

\usepackage{mathrsfs, mathtools}

\let\emph\undefined
\newcommand{\emph}[1]{\textsl{#1}}

\usepackage{hyperref}

\numberwithin{equation}{section}
\usepackage{mathtools}
\mathtoolsset{showonlyrefs}
\numberwithin{equation}{section}
\newtheoremstyle{style1}
  {13pt}
  {13pt}
  {}
  {}
  {\normalfont\bfseries}
  {.}
  {.5em}
  {}
\theoremstyle{style1}

\newtheorem{definition}[equation]{Definition}
\newtheorem{example}[equation]{Example}

\usepackage{xcolor}
\definecolor{Blue}  {rgb} {0.282352,0.239215,0.803921}
\definecolor{Green} {rgb} {0.133333,0.545098,0.133333}
\definecolor{Red}   {rgb} {0.803921,0.000000,0.000000}
\definecolor{Violet}{rgb} {0.580392,0.000000,0.827450}

\newcounter{commentnumber}

\DeclareMathOperator{\Ad}{Ad}

\DeclareMathOperator{\im}{im}

\DeclareMathOperator{\e}{e}

\newcommand\Hilb{{\mathscr{H}\mathsf{ilb}}}

\newcommand{\Bord}[1]{{\mathscr{B}\mathsf{ord}_2^{\hspace{1pt} #1}}}

\newcommand{\ctimes}{ \circlearrowleft }

\newcommand\Cb            {\mathbb{C}}

\newcommand\Rb            {\mathbb{R}}
\newcommand\Zb            {\mathbb{Z}}

\newcommand\Ac            {\mathcal{A}}

\newcommand\Cc            {\mathcal{C}}
\newcommand\Dc            {\mathcal{D}}

\newcommand\Fc            {\mathcal{F}}
\newcommand\Hc            {\mathcal{H}}

\newcommand\Pc            {\mathcal{P}}

\newcommand\Sc            {\mathcal{S}}

\newcommand\Wc            {\mathcal{W}}
\newcommand\Zc            {\mathcal{Z}}
\newcommand\funZ            {\mathcal{Z}}

\usepackage{tocloft}

\newtheoremstyle{style2}
  {13pt}
  {13pt}
  {\slshape}
  {}
  {\normalfont\bfseries}
  {.}
  {.5em}
  {}

\theoremstyle{style2}

\newtheorem{lemma}[equation]{Lemma}
\newtheorem{theorem}[equation]{Theorem}
\newtheorem{proposition}[equation]{Proposition}

\usepackage{tikz}
\usetikzlibrary{matrix,arrows,decorations.pathmorphing}
\usepackage{tikz-cd}

\usepackage{tikz}
\newcommand{\tr}{\mathrm{tr}}

\newcommand{\R}{\mathbb{R}}
\newcommand{\C}{\mathbb{C}}
\newcommand{\Z}{\mathbb{Z}}

\newcommand{\N}{\mathbb{N}}

\newcommand{\Bscr}{\mathscr{B}}

\newcommand{\sfH}{{\mathsf{H}}}

\newcommand{\End}{\mathsf{End}}
\newcommand{\Aut}{\mathsf{Aut}}
\newcommand{\Out}{\mathsf{Out}}
\newcommand{\Com}{\mathsf{Com}}
\newcommand{\Cl}{C\ell^2}
\newcommand{\Hom}{\mathsf{Hom}}

\newcommand{\Tr}{\text{Tr}}
\newcommand{\id}{\text{id}}

\newcommand{\iu}{\mathrm{i}}
\newcommand{\dd}{\mathrm{d}}

\newcommand{\Bun}{\mathsf{Bun}}
\newcommand{\BunG}{\mathsf{Bun}_G}

\begin{document}

\renewcommand{\thefootnote}{\fnsymbol{footnote}}

\begin{flushright}
\small
{\sf EMPG--19--17} \\
\textsf{Hamburger Beitr\"age zur Mathematik Nr. 795}\\
{\sf ZMP--HH/19--11} \\
\end{flushright}

\vspace{10mm}

\begin{center}
	\textbf{\LARGE{Symmetry defects and orbifolds \\[3mm] of two-dimensional Yang-Mills theory}}\\
	\vspace{1cm}
	{\large  Lukas M\"uller$^{\,a,}$}\footnote{Email: \ {\tt lm78@hw.ac.uk}}, \ \ {\large Richard J. Szabo$^{\,a,b,}$}\footnote{Email: \ {\tt r.j.szabo@hw.ac.uk}}  \ \ and \ \ {\large L\'{o}r\'{a}nt Szegedy$^{\,c,}$}\footnote{Email: \ {\tt lorant.szegedy@uni-hamburg.de}}

\vspace{5mm}

{\em $^a$ Department of Mathematics\\
Heriot-Watt University\\
Colin Maclaurin Building, Riccarton, Edinburgh EH14 4AS, U.K.}\\
and {\em Maxwell Institute for Mathematical Sciences, Edinburgh, U.K.}\\
and {\em The Higgs Centre for Theoretical Physics, Edinburgh, U.K.}
\\[3mm]

{\em $^b$ Dipartimento di Scienze e Innovazione
  Tecnologica}\\ {\it Universit\`a del Piemonte Orientale}\\
{\it Viale T. Michel 11, 15121 Alessandria, Italy}\\ and
{\it Arnold--Regge Centre, Via P. Giuria 1, 10125 Torino, Italy}
\\[3mm]

{\em $^c$ Fachbereich Mathematik\\ 
Universit\"at Hamburg\\
Bundesstra\ss e 55, D\,--\,20146 Hamburg, Germany}
\\[7pt]

\end{center}

\vspace{1cm}

\begin{abstract}
\noindent
We describe discrete symmetries of two-dimensional Yang-Mills theory with gauge group $G$ associated to 
outer automorphisms of $G$, and their corresponding defects. We show
that the gauge theory partition function with
defects can be computed as a path integral over the space of twisted
$G$-bundles, and calculate it exactly. We argue that 
its weak-coupling limit computes the
symplectic volume of the moduli space of flat twisted $G$-bundles on a surface. 
Using the defect network approach to generalised orbifolds, we gauge
the discrete symmetry and
construct the corresponding orbifold theory, which is again
two-dimensional Yang-Mills theory but with gauge group given by an
extension of $G$ by outer automorphisms. With the help of the orbifold
completion of the topological defect bicategory of two-dimensional
Yang-Mills theory, we describe the reverse orbifold using a Wilson
line defect for the discrete gauge symmetry.
We present our results using two complementary approaches: in the
lattice regularisation of the path integral, and in the functorial approach to
area-dependent quantum field theories with defects via regularised Frobenius algebras.  
\end{abstract}

\newpage

\tableofcontents

\bigskip

\setcounter{footnote}{0}
\renewcommand{\thefootnote}{\arabic{footnote}}

\section{Introduction}
\label{sec:intro}

Defects are impurities, discontinuities, boundary conditions, or other ways of modifying a 
quantum field theory which are localised on submanifolds of various
codimensions of the spacetime. In general,
defects can meet on other defects of higher codimensions. 
They are an important tool in the study of non-perturbative features
in quantum field theories and condensed matter systems, and of
symmetry protected topological phases.
The structure of the collection of all defects is closely
related to higher categories, see \cite{Davydov:2011dt,Kapustin} for a discussion in the case of topological field
theories.   
Defects are extended observables in quantum field theories, so that one can compute a partition function
or correlation function in the presence of a defect network. A prominent class of examples 
are the Wilson line observables in gauge theories.   

A particularly simple class of defects are topological defects which can be continuously deformed
without changing any physical observables. Two topological defects can be brought close together
such that their contribution to the partition function can be effectively described by another defect. This defines an operation
for topological defects called fusion. 
A defect $\mathtt{D}$ is called invertible 
if there exists another defect $\mathtt{D}^{-1}$ such that the fusion of  $\mathtt{D}$ with $\mathtt{D}^{-1}$ is the trivial 
defect $1$; every quantum field theory admits a trivial defect $1$ 
which does not change the value of the partition function.
       
Invertible topological defects of codimension 
one, or domain walls, are closely related to symmetries of the field theory~\cite{Fuchs:2015sym}. Given a symmetry we can construct a topological
domain wall which acts on fields passing through it by applying the symmetry. 
On the other hand, one can recover the action of the symmetry on fields by wrapping the defect around
a field insertion. The description of symmetries and their
corresponding background gauge fields via topological domain walls makes it possible
to describe their action on other quantities of the quantum field theory such as boundary conditions 
or other (non-topological) defects~\cite{Tachikawa:2017g}.\footnote{This description appears in much earlier literature on conformal field theory, where it is an immediate consequence of the condition that the topological defect commutes with both copies of the Virasoro algebra (see e.g.~\cite{Petkova}).}

In this paper we illustrate the relationship between defects and symmetries for a simple 
class of symmetries of two-dimensional Yang-Mills theory. Yang-Mills
theory on a Riemann surface has a long and rich history as an exactly
solvable quantum gauge theory which is the first example of a
non-abelian gauge theory that can be reformulated
as a (topological) string theory (see~\cite{Cordes:1995ym} for a
review). Mathematically it has served as a tool for studying the
topology of various moduli spaces of interest in geometry and
dynamical systems, such as the
moduli spaces of flat connections~\cite{Atiyah:1982fa,Witten:1991gt,Witten:1992xu}, the Hurwitz moduli spaces of branched
coverings~\cite{Cordes:1995ym,Kokenyesi:2016cun}, and the 
principal moduli spaces of holomorphic
differentials~\cite{Griguolo:2004uz,Griguolo:2004jp}. In the following
we will study how some of these features are modified in the
presence of domain walls, which in two dimensions can be thought of as
symmetry twist branch cuts on the surface~\cite{Barkeshli:2014cna}, corresponding to outer automorphisms of
the gauge group.\footnote{Similar defects are constructed
  in~\cite{Fuchs:2015bra} in the context of
  three-dimensional Dijkgraaf-Witten theories, and have also been
  proposed to have physical realisations in certain topological states
of matter, see for example~\cite{Barkeshli:2012pr}.}

Recently it was shown that defects in area-dependent two-dimensional quantum field theories
can be studied in terms of `regularised Frobenius algebras' and their
bimodules~\cite{Runkel:2018aqft}. Let $G$ be a compact semi-simple Lie group.
The square-integrable functions on $G$ form the regularised Frobenius algebra underlying the two-dimensional
Yang-Mills theory with gauge group $G$. In \cite{Runkel:2018aqft} invertible defects are constructed algebraically  
from outer automorphisms of $G$. For example, the only non-trivial
outer automorphism of $G=SU(n)$ for $n>2$ is the complex
conjugation automorphism $g\longmapsto\bar g$; further examples of groups with non-trivial outer automorphisms
are displayed in Table~\ref{table:OutG}.
\begin{center}
\begin{table}[htb]
\begin{center}
	\begin{tabular}{|c||c|c|c|c|}
		\hline {$G$} & {$SU(n)$ , $n>2$}& {$SO(2n)$ , $n>4$} & {$SO(8)$} &
                                                                     {$E_6$}
		\\[4pt] \hhline{|=|=|=|=|=|}
	$\Out(G)$ & $\Zb_2$ & $\Zb_2$ & $S_3$ & $\Zb_2$
		\\\hline
	\end{tabular}
\end{center}
\caption{\small The compact connected simple
  Lie groups $G$ with non-trivial outer automorphism groups
  $\Out(G)$. Here $\Zb_2$ is the abelian cyclic group of order~two and $S_3$
  is the non-abelian symmetric group of degree~three.}
\label{table:OutG}
\end{table}
\end{center}

One motivation for the present investigation is to give a physical 
interpretation of these defects, which is the content of the first part
of this paper (Section~\ref{Sec: Physics}). We show in Section~\ref{Sec:Symmetry} that for an outer automorphism
$\varphi\colon G\longrightarrow G$ there is a symmetry of two-dimensional Yang-Mills theory 
sending a gauge field described by a principal bundle with connection to the associated $G$-bundle
for the group homomorphism $\varphi$ with its induced connection. We
proceed to show that the partition function in the presence of a network of the associated defects on a closed oriented
surface $\Sigma$ can 
be computed as a path integral over the space of `twisted bundles'.
Let $\Bun_{G}^\nabla(\Sigma)$ be the 
space of principal
$G$-bundles with connection on $\Sigma$, $\Out(G)$ the (finite) group of outer
automorphisms of $G$
and $G\rtimes \Out(G)$ the semi-direct product of groups.\footnote{Generally, $\Out(G)$ is defined as the quotient of the group 
of automorphisms $\Aut(G)$ by the normal subgroup of inner automorphisms, and hence is not 
a subgroup of $\Aut(G)$. However, the construction of outer automorphisms from
symmetries of Dynkin diagrams, reviewed in Section~\ref{Sec:Lattice}, gives an
embedding of $\Out(G)\xhookrightarrow{ \ ~~ \ } \Aut(G)$. We implicitly use this embedding 
throughout the paper, for example when defining the semi-direct product $G\rtimes \Out(G)$.}
There is a natural group homomorphism $G\rtimes \Out(G)\longrightarrow \Out(G)$ which 
induces a map $\Bun_{G\rtimes \Out(G)}^\nabla(\Sigma)\longrightarrow \Bun^\nabla_{\Out(G)}(\Sigma)$.     
In Section~\ref{Sec:Defect} we construct an $\Out(G)$-bundle $D\longrightarrow\Sigma$ from the defect network and
define the space of $D$-twisted $G$-bundles with connection on $\Sigma$
as the fibre of $D$ for the map
$\Bun_{G\rtimes \Out(G)}^\nabla(\Sigma)\longrightarrow \Bun^\nabla_{\Out(G)}(\Sigma)$. 
Concretely, a twisted bundle on $\Sigma$ is given 
by a $G\rtimes \Out(G)$-bundle together with a gauge transformation from the induced 
$\Out(G)$-bundle to $D$ which specifies where the branch cuts on
$\Sigma$ are located.       

We present exact calculations of the partition function in the presence of defect networks 
using the lattice regularisation of two-dimensional Yang-Mills
theory in Section~\ref{Sec:Lattice}. This shows in particular 
that the defects corresponding to the symmetry induced by $\varphi$ agree with the 
defects constructed
in~\cite{Runkel:2018aqft}.
The weak-coupling limit of two-dimensional Yang-Mills theory can be used to compute 
the symplectic volume of 
the moduli space of flat connections on $\Sigma$~\cite{Witten:1991gt}. Similarly, we argue in 
Section~\ref{Sec:Zero} that the weak-coupling limit of two-dimensional Yang-Mills theory
in the presence of defects computes the symplectic volume of the moduli space of 
flat twisted bundles on $\Sigma$, introduced for example in~\cite{Meinrenken:Convexity}.

The description of symmetries via defects is particularly convenient in the context
of orbifold theories.\footnote{Orbifolds of two-dimensional conformal field theories
  with respect to an outer automorphism group symmetry were studied
  long ago, see for example~\cite{Schellekens:1990xy,Birke:1999ik}.}
The partition function of the orbifold theory
can be computed by evaluating the original theory on a sufficiently fine defect network 
labeled with symmetry defects.\footnote{By `sufficiently fine' we mean
a defect network that can be obtained as the dual graph of a
triangulation of $\Sigma$.} This requires the introduction of point defects 
--- called `junction fields' ---
on which three defect lines can meet.
We use the defect approach to orbifolds~\cite{FFSR:2009orb,Carqueville:2012orb,Brunner:2013orb,Carqueville:2017orb} 
in Section~\ref{Sec:Orbifold}
to show that the orbifold theory is the Yang-Mills theory on $\Sigma$ based on the gauge group 
$G\rtimes \Out(G)$. An advantage of the defect approach to orbifolds is that it
also works for suitable non-invertible defects. We use this generalised orbifold 
construction to show in Section~\ref{sec:reverseorb} that the orbifold of the Yang-Mills theory with gauge 
group $G\rtimes \Out(G)$, with respect to the Wilson line defects induced from 
the regular 
representation $L^2(\Out(G))$ of $\Out(G)$,
is the Yang-Mills theory on $\Sigma$ with gauge group
$G$; these defects are invertible, and hence correspond to symmetry defects, only when
$\Out(G)$ is an abelian group. This may be thought of as a duality between these two quantum Yang-Mills theories with defects.

The second part of this paper (Section~\ref{Sec:Funct}) is concerned with a mathematically 
rigorous formulation of the 
orbifold construction performed in the first part. We aim to show that this is a good illustration of the power of the approach to 
two-dimensional area-dependent quantum field theories via functorial field theories
and regularised Frobenius algebras. A regularised 
Frobenius algebra consists of a Hilbert space $A$ equipped with families
of bounded linear operators $\mu_a\colon A\otimes A\longrightarrow A$, $\eta_a\colon \C 
\longrightarrow A$, $\Delta_a \colon A\longrightarrow A\otimes A$ and $\varepsilon_a \colon A \longrightarrow \C$
parameterised by a positive real number $a\in \R_{>0}$. These maps are required to be continuous 
in an appropriate sense and to satisfy parameterised versions of the
usual relations for 
Frobenius algebras. We give the full definition in Section~\ref{Sec:State-Sum-Intro}. 

The example relevant to this paper is the Hilbert space of square-integrable functions 
$A=L^2(G)$ on a compact semi-simple Lie group $G$, which becomes a regularised Frobenius algebra 
via the morphisms
\begin{align}
\hspace{-5mm} \eta_a(1)& =
           \sum_{\alpha\in\widehat{G}}\,\dim\alpha\,\exp\Big(-a\,\frac{C_2(\alpha)}{2}\Big)\,
           \chi_\alpha \ , 
& \mu_a(f\otimes g)&= \eta_a(1)\ast(f\ast g) \ , \nonumber \\[4pt]
\hspace{-5mm} \varepsilon_a(f)&= \big(\eta_a(1)\ast f\big)(1) \ , & \Delta_a f &=
                                             \Delta\big(\eta_a(1)\ast f\big)
                                             \  , \label{eq:L2G-RFA-3}
\end{align}   
where the sum runs over all isomorphism classes $\alpha$ of
irreducible representations of $G$ of dimension $\dim\alpha$,
$\chi_\alpha:G\longrightarrow\C$ denotes the corresponding character and $C_2(\alpha)$ is the 
value of the quadratic Casimir operator $C_2$ in the representation
$\alpha$. 
Here
\begin{align}
(f* g)(x)= \int_{G}\, \dd y \ f\big(x\,y^{-1}\big)\,g(y) \qquad \mbox{and}
           \qquad (\Delta f )(x,y)& = f(x\,y)
\end{align}
are the usual convolution product and coproduct on $L^2(G)$.
The regularised Frobenius algebra $L^2(G)$ is the input for the state sum construction of two-dimensional Yang-Mills theory
in~\cite{Runkel:2018aqft}, which makes the lattice regularisation of
Section~\ref{Sec: Physics} precise. 
The centre of $L^2(G)$ is the commutative regularised Frobenius algebra of class functions $C\ell^2(G)$ on $G$, which describes two-dimensional Yang-Mills theory without defects.

In this approach, defects between different Yang-Mills theories
correspond to dualisable bimodules between the regularised Frobenius algebras $L^2(G)$
and $L^2(G')$. Such a bimodule consists of a Hilbert space $M$ together with a family of maps 
$\rho_{a,b}\colon L^2(G) \otimes M \otimes L^2(G')\longrightarrow M$ parameterised 
by two positive real numbers $a,b \in \R_{>0}$, satisfying a parameterised version of the
usual bimodule relations; we refer again to Section~\ref{Sec:State-Sum-Intro} for more details. 
The defect described in the first part of the paper corresponds to the bimodule 
$L_\varphi= L^2(G)$ with action twisted by the outer automorphism $\varphi$ as
\begin{align}
\rho_{a,b}\colon L^2(G)\otimes L^2(G) \otimes L^2(G) &\longrightarrow
                                                       L^2(G) \ , \\
f\otimes h \otimes g &\longmapsto \mu_a\big(f\otimes\mu_b(h\otimes \varphi^*g)\big) \ .
\label{eq:twisted-bimodule-action}
\end{align}     

These observations are the starting point for the content of second part of the paper.
After a brief review of the state sum construction of area-dependent quantum field 
theories from regularised Frobenius algebras and their bimodules in
Section~\ref{Sec:State-Sum-Intro}, we make the 
mathematical structure of topological defects in two-dimensional Yang-Mills theories 
precise by introducing an idempotent complete bicategory of topological defects
in Section~\ref{Sec:Defect-Bicategory}.
The input for a generalised orbifold construction in this framework is a strongly
separable symmetric Frobenius algebra in an endomorphism category inside the topological defect 
bicategory. In Section~\ref{Sec:bimodule-from-Hopf} we construct such a Frobenius 
algebra from 
the defects which is isomorphic to a Frobenius algebra built from a bimodule
structure on the square-integrable functions $L^2(G\rtimes \Out(G))$ on the semi-direct 
product $G\rtimes \Out(G)$. In Section~\ref{Sec:Orbifolding-Out-G} we compute the corresponding 
orbifold theory using the state sum construction and show, using the abstract
orbifold completion of the defect bicategory \cite{Carqueville:2012orb}, that the backwards orbifold 
can be performed using Wilson line defects in Section~\ref{Sec:Backwards-Orbifold}.      
                       
\subsection*{Acknowledgments}

We thank Sanjaye Ramgoolam, Ingo Runkel and Lennart Schmidt for
helpful discussions and correspondence.
This work was supported by the COST Action MP1405 ``Quantum Structure of Spacetime'', funded by the European Cooperation in Science and Technology (COST).
The work of L.M. was supported by the Doctoral Training Grant ST/N509099/1 from the UK Science and Technology Facilities Council (STFC).
The work of R.J.S. was supported in part by the STFC Consolidated
Grant ST/P000363/1 ``Particle Theory at the Higgs Centre'' and the
Istituto Nazionale di Fisica Nucleare (INFN, Torino).
L.S. gratefully acknowledges the Max Planck Institute for Mathematics
in Bonn for financial support and hospitality. 
L.S. is partially supported by the DFG Research Training Group 1670 “Mathematics
Inspired by String Theory and Quantum Field Theory”.
 
\section{Discrete symmetries and defects of Yang-Mills theory}
\label{Sec: Physics}

Let $G$ be a compact semi-simple Lie group with Lie algebra
$\mathfrak{g}$, and 
$\varphi\colon G \longrightarrow G$
an outer automorphism of $G$. In this section we construct a groupoid homomorphism
$\varphi \colon \BunG^\nabla(\Sigma) \longrightarrow
\BunG^\nabla(\Sigma)$ on the topological groupoid of principal $G$-bundles with connection for every closed
oriented surface $\Sigma$. We show in
Section~\ref{Sec:Symmetry} that the Yang-Mills action functional on $\Sigma$ is invariant under this 
transformation, so that $\varphi$ defines a symmetry of the gauge theory. In Section~\ref{Sec:Defect} we study the
corresponding defects in terms of twisted bundles, and calculate the
partition functions exactly
using a lattice
regularisation of the quantum gauge theory in Section~\ref{Sec:Lattice}.  
The partition function for 
a given defect configuration localises in the weak-coupling limit onto
the moduli space of flat twisted $G$-bundles, similarly to the
untwisted case~\cite{Witten:1991gt}. In this limit the
partition function computes the symplectic volume of this moduli space, defined for example in~\cite{Meinrenken:Convexity}. We will
use the combinatorial quantisation of two-dimensional Yang-Mills theory to make a precise conjecture for this volume in some simple cases in Section~\ref{Sec:Zero}. In Section~\ref{Sec:Orbifold} 
we compute the corresponding orbifold theory, and subsequently the
reverse orbifold theory in Section~\ref{sec:reverseorb}. 

\subsection{Description of the \texorpdfstring{$\Out(G)$}{Out(G)}-symmetry}\label{Sec:Symmetry} 

Let $\pi:P\longrightarrow \Sigma$ be a 
principal $G$-bundle on $\Sigma$, and let
$\Ad(P)=(P\times\mathfrak{g})/G$ be the vector bundle on $\Sigma$ with
fibre $\mathfrak{g}$, where $G$ acts on $P$ by the
principal bundle action and on $\mathfrak{g}$ by the adjoint
action. Let $A \in
\Omega^1(P;\mathfrak{g})$ be a connection on $P$; its curvature $F_A=\dd A+A\wedge A$ is a
two-form on $\Sigma$ valued in the adjoint bundle associated with $P$:
$F_A\in\Omega^2(\Sigma;\Ad(P))$. The group of gauge transformations is
the automorphism group $\Aut(P)$ of $P$ consisting of $G$-equivariant
diffeomorphisms $g:P\longrightarrow P$ with $\pi\circ g=\pi$. 

We define a new bundle $
\varphi (P)$ with the same underlying total space $P$ and projection $\pi$, but 
with $G$-action $P\times G\longrightarrow P$ modified by pre-composing with $\varphi^{-1}$.
This can also be regarded as the induced $G$-bundle $P\times_\varphi G$ for the 
Lie group automorphism $\varphi\colon G \longrightarrow G$.
By differentiating $\varphi$ at 
the identity we get a Lie algebra automorphism $\varphi_*\colon \mathfrak{g}\longrightarrow
\mathfrak{g}$. 
The symmetry $\varphi$ acts on the connection $A$ 
by mapping it to $\varphi_*(A)$ where 
$\varphi_* \colon \mathfrak{g}\longrightarrow \mathfrak{g}$ acts only on             
the Lie algebra part of the one-form $A\in \Omega^1(P; \mathfrak{g})$;
this is the induced connection on the bundle $P\times_\varphi G$, see
for example~\cite[Section~1]{Freed:1995cs}. 

We use a local trivialisation to show that this is again a principal 
$G$-bundle with connection.
Let $\{{U}_i\}$ be an open cover of $\Sigma$ such that $(P,A)$ can be 
described by transition functions 
$g_{ij}\colon U_{ij}\longrightarrow G$ on overlaps $U_{ij}:={U}_i\cap{U}_j $ and local
one-forms $A_i \in \Omega^1({U}_i; \mathfrak{g})$. On the overlaps
$U_{ij}$ the one-forms $A_i$
are required
to satisfy 
\begin{align}\label{Eq: Condition connection}
A_j= \Ad_{g_{ij}}(A_i)+ g_{ij}^*\theta = g_{ij}^{-1}\,A_i\, g_{ij} + g^{-1}_{ij}\,\dd g_{ij}
\end{align}
where $\theta$ is the Maurer-Cartan one-form on $G$ and the second equality 
holds for matrix Lie groups. 
Under the automorphism, $g_{ij}$ is mapped to $ \varphi \circ g_{ij}\colon U_{ij}
\longrightarrow G$. We check that $\varphi_*(A_i)$ satisfy 
\eqref{Eq: Condition connection} with respect to the new transition functions:
\begin{align}
\Ad_{\varphi\circ g_{ij}}\big(\varphi_*(A_i)\big)+(\varphi\circ g_{ij})^*\theta &= \varphi_*\big(\Ad_{g_{ij}}(A_i)\big)+(\varphi\circ g_{ij})^*\theta \\[4pt]
&= \varphi_*\big(\Ad_{g_{ij}}(A_i)+ g_{ij}^*\theta\big) \\[4pt]
&= \varphi_*(A_j) \ ,
\end{align}
where the first equality follows from differentiating the equality
$\varphi(g^{-1})\,\varphi(\,\cdot\,)\,\varphi(g)= \varphi\big(g^{-1}\,(\,\cdot\,)\, g\big)$ of group 
automorphisms of $G$ and the second equality from the definition of
the Maurer-Cartan one-form.
The groupoid homomorphism acts on a gauge transformation $\xi_i:U_i \longrightarrow G$ by
composition with $\varphi$.

In the case of an inner automorphism of $G$, this would just describe the 
action of a global gauge transformation on the physical fields. 
For later use, we note that the action on the parallel transport in $P$
described by an element $g\in G$ with respect to a given local 
trivialisation is given by applying $\varphi$
to $g$.\footnote{This can be seen by applying the chain rule to the function 
$\varphi \circ g\colon [0,1]\longrightarrow G$, where $g\colon [0,1]\longrightarrow G$
is a solution to the differential equation describing the parallel transport in $P$.}

Let us now equip the smooth oriented surface $\Sigma$ with a
Riemannian metric such that $\Sigma$ has total area 
\begin{align}
a:=\int_\Sigma\, \dd\mu \ ,
\end{align}
where $\dd\mu=\star\,1$ is the Riemannian measure defined by the
corresponding Hodge duality operator $\star\,$. The metric induces the Hodge operator
$\star\,\colon \Omega^2(\Sigma;\Ad(P))\longrightarrow
\Omega^0(\Sigma;\Ad(P))$ acting only on the differential form part. We also equip the Lie
algebra $\mathfrak{g}$ with an invariant quadratic form, which we may
assume without loss of generality to be a suitable multiple, as described
in~\cite{Witten:1991gt}, of the Killing form on $\mathfrak{g}$
denoted by $\text{Tr}_{\mathfrak{g}}$. Then Yang-Mills theory on $\Sigma$ is
defined by the $\Aut(P)$-invariant action functional of
$(P,A)\in\BunG^\nabla(\Sigma)$ given by 
\begin{align}\label{eq:YMaction}
S_{\textrm{\tiny YM}}(P,A)\coloneqq \frac{1}{4e^2}\,\int_\Sigma\,
  \text{Tr}_{\mathfrak{g}}\big(F_A\wedge \star\,F_A \big) \ ,
\end{align}
where $e$ is the gauge coupling constant. 

The Yang-Mills action functional transforms under $\varphi$ according to
\begin{align}
\text{Tr}_{\mathfrak{g}}\big(F_A\wedge \star\, F_A\big) \xrightarrow{ \ \varphi \ }
  \text{Tr}_{\mathfrak{g}}\big(\varphi_*(F_A)\wedge \star\, \varphi_* (F_A) \big) \ ,
\end{align}
 where $\varphi_*$ again acts only on the Lie algebra part. 
The invariance of the action functional $S_{\textrm{\tiny YM}}(P,A)$ then follows from the fact that the Killing form is 
preserved by Lie algebra automorphisms. 
This shows that $\varphi$ induces a symmetry of the gauge
theory at the classical level. 

However,
this does not necessarily extend to a symmetry at the quantum
level. To define the quantum gauge theory, we note that the tangent
space at any point $(P,A)\in\BunG^\nabla(\Sigma)$ can be identified with
$\Omega^1(\Sigma;\Ad(P))$, and thus an
$\Aut(P)$-invariant symplectic form on $\BunG^\nabla(\Sigma)$ can be
defined by~\cite{Atiyah:1982fa}
\begin{align}\label{eq:SymplecticForm}
\omega(A_1,A_2)=\frac1{4\pi^2}\, \int_\Sigma\, \text{Tr}_{\mathfrak{g}}\big(A_1\wedge
  A_2\big) \ ,
\end{align}
for any two
$\Ad(P)$-valued one-forms $A_1$ and $A_2$ on $\Sigma$. This formally
induces an $\Aut(P)$-invariant symplectic measure on the
infinite-dimensional space
$\BunG^\nabla(\Sigma)$ which
we denote by $\mathscr{D}(P,A)$. 
Two-dimensional quantum Yang-Mills theory is then defined by 
the partition function which is given as the formal Euclidean path integral
\begin{align}\label{Eq: Partition function}
Z_{\textrm{\tiny YM}}\big(\Sigma,G,e^2\,a\big) := \int_{\BunG^\nabla
  (\Sigma)}\, \mathscr{D}(P,A) \ \exp\big(- S_{\textrm{\tiny YM}}(P,A)\big) 
  \ .
\end{align} 
The partition function
is invariant under area-preserving diffeomorphisms of the Riemann
surface $\Sigma$, and so depends on $e$ and the metric of $\Sigma$ only through the combination
$e^2\,a$~\cite{Witten:1991gt,Cordes:1995ym}. It is therefore possible
to set $e=1$ without loss of generality and consider all amplitudes of the gauge theory as functions of the
area $a$. In this sense two-dimensional Yang-Mills theory is a mild
variant of a topological field theory that is an
example of an `area-dependent quantum field theory'.

Now the same arguments used to establish invariance of the classical
Yang-Mills action functional under the symmetry $\varphi$ show that
the symplectic form \eqref{eq:SymplecticForm} is preserved by
$\varphi$. Thus the partition function \eqref{Eq: Partition function}
is {\it formally} invariant under the symmetry $\varphi$. However, the
definition of the formal path integral \eqref{Eq: Partition function}
requires a suitable regularisation to make it mathematically well-defined, and it may
be that there is no regularisation which preserves the symmetry; in
such a case the symmetry is anomalous and the partition function is not
invariant. We will show in Section~\ref{Sec:Lattice} below, using the
lattice regularisation of two-dimensional Yang-Mills theory, that
indeed the quantum gauge theory is also invariant under the symmetry
induced by the outer automorphism $\varphi:G\longrightarrow
G$. Different prescriptions for defining the path integral in
\eqref{Eq: Partition function} will differ by a
renormalisation ambiguity depending on the topology and area of
$\Sigma$ as~\cite{Witten:1991gt,Witten:1992xu}
\begin{align}\label{eq:ambiguity}
\Delta S =
\upsilon_1\, \chi(\Sigma)+\upsilon_2\, e^2\, a
\end{align}
 for arbitrary constants
$\upsilon_1,\upsilon_2\in\R$, where $ \chi(\Sigma)$ is the Euler
characteristic of $\Sigma$; this respects the invariance under
area-preserving diffeomorphisms and just multiplies the partition
function by a constant factor $\exp(-\Delta S)$. The parameters
$\upsilon_1,\upsilon_2$ depend only on the gauge group $G$ and
the renormalisation scheme, but not on the area or topology of the
surface $\Sigma$.

At this stage though we can already see how 
the symmetry acts on the Hilbert space of 
wavefunctions. 
The quantum Hilbert space of the gauge theory on a Cauchy
circle $S^1$ in $\Sigma$ consists of gauge-invariant functions from 
the collection of principal bundles with connection on $S^1$ to $\C$. 
The only gauge-invariant quantity that one can associate to a $G$-bundle with 
connection over $S^1$ is 
the conjugacy class of its holonomy around the circle. Hence the state
space is given by the Hilbert space
\begin{align}
Z_{\textrm{\tiny YM}}\big(S^1,G\big) = C\ell^2(G):=L^2(G)^{\Ad(G)}
\end{align}
 of class functions on $G$. 
The collection of characters $\chi_\alpha$ of
unitary irreducible 
representations $\alpha$ of $G$ provide a natural basis for the
state space. The symmetry acts unitarily on this Hilbert space by sending a class function $f\colon G \longrightarrow \C$ to the function $f\circ\varphi^{-1}$.

The Hamiltonian $H_{\textrm{\tiny YM}}$ of the gauge theory associated
to any foliation of the surface $\Sigma$ is given in terms of 
the quadratic Casimir operator $C_2$ 
 by~\cite{Cordes:1995ym,Witten:1992xu}
\begin{align}\label{eq:Hamiltonian}
H_{\textrm{\tiny YM}}=\frac{e^2}2\, L\, C_2+e^2\,L\,\upsilon_2 \ ,
\end{align}
where $L$ is the length of the Cauchy circle. The Hamiltonian
operator \eqref{eq:Hamiltonian} is the generator of time
translations. The time evolution operator $\exp(-T\,H_{\textrm{\tiny YM}})$ is parameterised by the
elapsed time $0\leq
T\leq \frac aL$, and its action on the
character basis is given by
\begin{align}
\exp(-T\,H_{\textrm{\tiny YM}})\,\chi_\alpha = \exp\bigg(-e^2\,L\,T\, \Big(
  \frac{C_2(\alpha)}2+\upsilon_2 \Big) \bigg) \ \chi_\alpha \ .
\end{align}        
The symmetry maps $\chi_\alpha$ to $\chi_{\varphi^*\alpha}$, and thus commutes with the time evolution operator since 
the invariance of the Killing form implies
$C_2(\alpha)=C_2(\varphi^*\alpha)$: The action of $\varphi$ on the
Casimir operator $C_2$ can be interpreted as a change
of basis in the Lie algebra corresponding to the Lie algebra automorphism
$\varphi_*$, which 
preserves the Killing form and hence maps dual coordinates 
to dual coordinates.
As a result, 
it only changes the choice of basis for the evaluation of the Casimir 
operator and leaves its value invariant.

\subsection{Defects and twisted bundles}\label{Sec:Defect}

Every symmetry of a quantum field theory comes with a corresponding invertible codimension one topological defect,
such that passing a field through the defect corresponds to the action of the 
symmetry on the field, as illustrated in Figure~\ref{Fig:Sym_Defect}. 
\begin{figure}[htb]
\small
\begin{center}
\begin{overpic}[scale=1]
{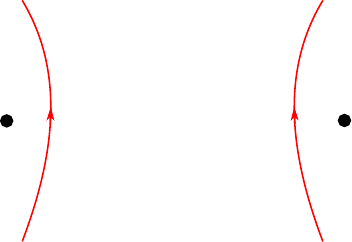}
\put(0,66){$\Psi$}
\put(157,66){$\varphi\cdot\Psi$}
\put(70,55){{\Large$\longleftrightarrow$}}
\end{overpic}
\end{center}
\caption{\small A defect corresponding to the symmetry $\varphi$,
  indicated by the directed lines. A field insertion $\Psi$ moving through 
the defect corresponds to the action of $\varphi$ on $\Psi$.}
\label{Fig:Sym_Defect}
\normalsize
\end{figure}
In general, different defects can join at lower-dimensional 
submanifolds. In the following we give a geometric description
of the partition function of two-dimensional Yang-Mills theory with
gauge group $G$ in the presence of an arbitrary network of defects
corresponding to the symmetry discussed in Section~\ref{Sec:Symmetry}.

A \emph{defect network} on the surface $\Sigma$
is determined by a triangulation of $\Sigma$, a choice of an orientation, and an 
element in $\Out (G)$ for every edge of the triangulation; the defect
network represents the dual triangulation. Since $\Out(G)$ is a finite
group, we require that around every 
vertex the product of all three elements in $\Out (G)$ is $\id_G$ (taking the
orientation into account as explained below). The
elements of $\Out(G)$ describe the types of defects corresponding to 
the edges. Edges labeled with $\id_G$ are interpreted as the absence of any 
defects (the trivial defects). This condition thus
expresses the fact that a defect with label $\varphi$ and a defect
with label $\varphi'$ can fuse at a codimension two junction to give a defect with label
$\varphi\,\varphi'$, and it allows us to express the labels as a
consistent configuration of domain walls throughout $\Sigma$.

Fix a defect network on $\Sigma$. It defines a principal
$\Out(G)$-bundle on $\Sigma$ as follows: 
For every face $i$ of the triangulation we pick an open neighbourhood
${U}_i$ which is only ``slightly bigger'' than the face.  
If there is an edge of the triangulation between ${U}_i$ and 
${U}_j$, we set the transition function from the left of the defect
to the right equal to the label of the corresponding edge. This is possible 
since we have picked an orientation for every edge. If the intersection of
${U}_i$ and ${U}_j$ is non-empty but there is no edge between
them, then we can go from ${U}_i$ to ${U}_j$ by passing through
a finite number of neighbourhoods for which there exists an edge for any
transition. The transition function $U_{ij}\longrightarrow\Out(G)$ is then uniquely fixed by the
cocycle condition. The consistency condition implies that this defines an 
$\Out(G)$-bundle, which we call $D$. Since $\Out(G)$ is a finite group,
$D$ carries a unique flat connection. The parallel transport from the 
centre of one face to an adjacent centre can be described by the action
of the element labeling the edge between the two faces. Hence passing to
the dual triangulation provides the holonomy description of the bundle
$D$ on $\Sigma$. 

The semi-direct product $G \rtimes \Out(G)$ is the group with 
elements $(g,\varphi)\in G\times \Out(G)$ and multiplication 
\begin{align}
(g,\varphi)\,{\textrm{\tiny$\bullet$}}\, (g',\varphi'):=\big(g\,\varphi(g'),\varphi\,
  \varphi'\big) \ . 
\end{align}
The projection onto the second factor 
$G \rtimes \Out(G) \longrightarrow \Out(G)$ is a group homomorphism and 
hence induces a map $\Bun_{G\rtimes \Out(G)}^\nabla(\Sigma) \longrightarrow
\Bun_{\Out(G)}^\nabla(\Sigma)$. A \emph{$D$-twisted $G$-bundle with connection}
is a $G \rtimes \Out(G)$-bundle with connection such that 
the induced $\Out(G)$-bundle is isomorphic to $D$. We denote by
$\Bun^\nabla_{G\downarrow D}(\Sigma)$ the space of $D$-twisted
$G$-bundles with connection on $\Sigma$.
 A more precise
but abstract definition states that the space of $D$-twisted 
$G$-bundles with connection is the homotopy fibre of the 
map $\Bun_{G\rtimes \Out(G)}^\nabla(\Sigma) \longrightarrow
\Bun_{\Out(G)}^\nabla(\Sigma)$. An object in this space also comes with
a particular choice of isomorphism which we suppress from the present
discussion. In the case of discrete gauge groups $G$, this definition reduces 
to the notion of relative bundles used in \cite{Fuchs:2013dw} to describe defects in three-dimensional Dijkgraaf-Witten theories.     

Twisted bundles implement the defects corresponding to the symmetry 
introduced in Section~\ref{Sec:Symmetry}. 
We can describe a twisted bundle with respect to the same open 
cover $\{U_i\}$ used to define the $\Out(G)$-bundle $D$. The transition functions for a
twisted bundle consist 
of pairs $(g_{ij}, \varphi_{ij})$, where $g_{ij}\colon
{U}_{ij}\longrightarrow G$ and 
$\varphi_{ij}$ are fixed to be the transition functions of $D$ 
(up to isomorphism). Hence $g_{ij}$ are the only free parameters we 
can choose. On triple overlaps $U_{ijk}=U_i\cap U_j\cap U_k$ the cocycle condition 
\begin{align}
(g_{ki}, \varphi_{ki})= (g_{kj},\varphi_{kj})\,{\textrm{\tiny$\bullet$}}\,(g_{ji},\varphi_{ji})
\end{align}  
implies that $g_{ki}= g_{kj}\,\varphi_{kj}(g_{ji})$, which can be interpreted
in the language of defects as the symmetry acting on the transition function
$g_{ji}$ when it passes through the defect labeled by $\varphi_{kj}$. 
A connection can be described locally by one-forms 
$A_i \in \Omega^1({U}_i;\mathfrak{g})$. 
This is the same data as required for a connection on a $G$-bundle.
However, its transformation rule is twisted. For instance, if all
$g_{ij}$ are trivial, then on $U_{ij}$ the connection one-forms
are required
to satisfy
\begin{align}
A_i = \Ad_{(1,\varphi_{ij})}( A_j) = {\varphi_{ij}}_*(A_j) \ .
\end{align} 
This is exactly the action of the symmetry on the connection one-forms
introduced in Section~\ref{Sec:Symmetry}. 
Gauge transformations of twisted bundles are described by elements $\xi_i \in 
\Omega^0(U_i;G)$ which, via the embedding $G\xhookrightarrow{ \ ~~ \ } G\rtimes \Out(G)$, induce a gauge transformation of the corresponding 
$G\rtimes \Out(G)$-bundle, or in other words gauge transformations of the 
$G\rtimes \Out(G)$-bundle which induce the trivial gauge transformation of
$D$.\footnote{Actually, we work with $G\rtimes \Out(G)$-gauge transformations that under the map 
$\Bun^\nabla_{G\rtimes \Out(G)}(\Sigma) \longrightarrow \Bun^\nabla_{\Out(G)}(\Sigma) $ relate the two identifications with $D$.}
A gauge transformation $\xi_i$ acts on the transition functions via 
\begin{align}
g_{ij}(x)\xrightarrow{ \ \xi_i \ } \xi_i(x)\,
g_{ij}(x)\,\varphi_{ij}\big(\xi_i(x)^{-1}\big)
\end{align}
for all $x\in U_{ij}$.

The Yang-Mills action functional for a twisted bundle is given by the
same formula \eqref{eq:YMaction} for the Yang-Mills action functional of the
corresponding $G\rtimes \Out(G)$-bundle. This action functional locally agrees 
with the Yang-Mills action functional for $G$-bundles. 

To define the corresponding quantum gauge theory, we note again that a
tangent vector to an arbitrary point $(P,A)\in\Bun^\nabla_{G\downarrow D}(\Sigma)$ is an
$\Ad(P)$-valued
one-form on $\Sigma$, where here $\Ad(P)$ is the vector bundle
associated to $P$ by the $G\rtimes \Out(G)$-action on the Lie algebra $\mathfrak{g}$. Given two tangent
vectors $A_1$ and $A_2$, we can define an $\Aut(P)$-invariant symplectic pairing by the same
formula \eqref{eq:SymplecticForm}. 
The partition function of two-dimensional quantum Yang-Mills theory in the background
of a defect network described by an $\Out(G)$-bundle $D$ on $\Sigma$
is then physically 
defined as the formal path integral
\begin{align}\label{Eq: Path integral definition defect}
Z_{\textrm{\tiny YM}}\big(\Sigma,G,e^2\,a;D\big) :=
  \int_{\Bun^\nabla_{G\downarrow D}(\Sigma)}\, \mathscr{D}(P,A) \ 
  \exp\big(- S_{\textrm{\tiny YM}}(P,A)\big) \ ,
\end{align}   
where the integration is taken over the space of $D$-twisted $G$-bundles
with connection up to gauge
transformations and the measure $\mathscr{D}(P,A)$ is induced by the symplectic 
form. This can be regarded as a part of the partition function of the
Yang-Mills theory with gauge group $G\rtimes \Out(G)$. We will come back to this point in 
Section~\ref{Sec:Orbifold}. 

\subsection{Combinatorial quantisation of defect Yang-Mills theory}\label{Sec:Lattice}

We will now study the symmetries and defects introduced in Sections~\ref{Sec:Symmetry} and~\ref{Sec:Defect}
using the lattice formulation of two-dimensional Yang-Mills theory, as
reviewed for example in~\cite{Witten:1991gt,Cordes:1995ym}. 
To discretise the calculation 
we fix a cell decomposition consisting of an embedded graph which covers the compact oriented surface $\Sigma$ into vertices $\Sigma^{\textrm{\tiny(0)}}$, edges $\Sigma^{\textrm{\tiny(1)}}$,
and faces $\Sigma^{\textrm{\tiny(2)}}$. As the faces are contractible, the only remainders of a principal $G$-bundle on
$\Sigma$ are 
its fibre over every vertex, which we can trivialise. A gauge transformation
is therefore described by a map $\xi:\Sigma^{\textrm{\tiny(0)}}\longrightarrow G$. A connection on the
bundle is described by its parallel transport, which is a group element $g_\gamma\in G$
for every edge $\gamma\in\Sigma^{\textrm{\tiny(1)}}$; if the edge $\gamma$
joins vertex $x$ to vertex $y$, then a gauge transformation $\xi$ acts on
$g_\gamma$ by $g_\gamma\longmapsto \xi_y\,g_\gamma\,\xi_x^{-1}$. The
curvature of a connection is a gauge-invariant map
$\mathcal{U}:\Sigma^{\textrm{\tiny(2)}} \longrightarrow G$. 
In the lattice formulation the only gauge-invariant quantity one can 
construct for a face $w\in \Sigma^{\textrm{\tiny(2)}}$ is the (conjugacy class of the) holonomy 
around the face:
\begin{align}
\mathcal{U}_w= \prod_{\gamma\in\partial w}\, g_\gamma \ ,
\end{align} 
where the product runs over
the boundary edges $\gamma$ of $w$; here we choose an ordering
for the multiplication of edges similarly
to~\cite{Runkel:2018aqft}. 

The path integral can now be defined as an 
integral over the product group $G^{\times|\Sigma^{\textrm{\tiny(1)}}|}$ with respect to its
Haar measure, induced from the normalised invariant Haar measure $\dd g$ on $G$, which is 
mathematically well-defined. 
We still need to understand what to integrate. For this, note that formally
we can rewrite the path integral involving an arbitrary local
functional $\mathcal{L}(P,A)$ of the gauge fields as
\begin{align}
\int_{\Bun_G^\nabla (\Sigma)} \, \mathscr{D}(P,A) \ 
  \exp\Big(-\int_\Sigma\,
  \mathcal{L}(P,A)\Big) 
= \int_{\Bun_G^\nabla (\Sigma)} \, \mathscr{D}(P,A) \ \prod_{w\in \Sigma^{\textrm{\tiny(2)}}} \,
  \exp\Big(-\int_{w}\, \mathcal{L}(P,A)\Big) \ .
\end{align}  
Hence the integrand is a product over the faces $w$ of the cell decomposition
of $\Sigma$.
We further introduce a local measure on the
discretisation by giving an area $a_w$ for every face 
$w$. The integration factor is a local function ${\mit\Gamma}(\mathcal{U}_w,e^2\,a_w)$
depending on the holonomy and area associated to $w$.
The correct choice for this local factor which computes the Yang-Mills partition function is~\cite{Migdal:1975re,Witten:1991gt} 
\begin{align}\label{eq:mitGammadef}
{\mit\Gamma}\big(\mathcal{U}_w,e^2\,a_w\big) \coloneqq \e^{-\upsilon_1} \, \sum_{\alpha\in\widehat{G}} \, \dim\alpha
  \ \chi_\alpha(\mathcal{U}_w) \, \exp\bigg(-e^2\,a_w\, \Big(
  \frac{C_2(\alpha)}2 +\upsilon_2\Big) \bigg) \ ,
\end{align}       
where the sum runs over all isomorphism classes of unitary irreducible
representations $\alpha$ of the gauge group $G$ of dimension
$\dim\alpha$ and with character $\chi_\alpha$, and $C_2(\alpha)$ is the value of the quadratic Casimir operator
$C_2$ in the
representation $\alpha$.
This factor describes the wavefunction of two-dimensional Yang-Mills theory
on a disk~\cite{Migdal:1975re}  which is determined by the heat kernel
corresponding to the
Hamiltonian \eqref{eq:Hamiltonian}. In general it involves the
constants $\upsilon_1,\upsilon_2\in\R$ from \eqref{eq:ambiguity}
depending on the renormalisation scheme, where we used
$ \chi(w)=1$.\footnote{The explicit expressions for the partition functions below may then be
  alternatively derived by using standard glueing
rules from topological field theory, as in~\cite{Runkel:2018aqft}, and the additivity of the Euler
characteristic under disjoint unions, together with the fact that
pairs of pants have Euler characteristic $-1$.}

The partition function on $\Sigma$ in this lattice regularisation is now defined as 
\begin{align}
Z_{\textrm{\tiny YM}}\big(\Sigma,G, e^2\,a\big) \coloneqq
  \int_{G^{\times|\Sigma^{\textrm{\tiny(1)}}|}} \
  \prod_{\gamma\in \Sigma^{\textrm{\tiny(1)}}} \, \dd g_{\gamma} \ 
  \prod_{w\in \Sigma^{\textrm{\tiny(2)}}} \, {\mit\Gamma}\big(\mathcal{U}_w,e^2\,a_w\big) \ ,
\end{align}
where $a=\sum_{w\in\Sigma^{\textrm{\tiny(2)}}} \, a_w$.
This partition function is independent of the chosen cell decomposition of 
$\Sigma$~\cite{Witten:1991gt} and hence agrees with its continuum limit where the lattice
discretisation becomes finer and finer:
the heat kernel defines a
renormalisation group-invariant amplitude on the faces so that the
partition function is invariant under subdivision of the lattice. This
feature is special to Yang-Mills theory in two dimensions, and
for this reason the lattice regularisation of the quantum gauge theory actually computes
the partition function \eqref{Eq: Partition function} exactly (up to
the undetermined constants $\upsilon_1$ and~$\upsilon_2$). 

At this stage we can come back
to the question of whether $\varphi\in\Out(G)$ induces a symmetry of
the gauge theory at the quantum level. 
Since the pullback of the Haar measure along $\varphi$ is invariant and 
normalised, it follows from the uniqueness of the Haar measure that the
integration measure is preserved under $\varphi$. The action on the group element 
associated to an edge is given by applying $\varphi$ to it, since this 
describes the action on the parallel transport. 
Hence the integration factor transforms as
\begin{align}
{\mit\Gamma}\big(\mathcal{U}_w,e^2\,a_w\big) \xrightarrow{ \ \varphi \ } & \e^{-\upsilon_1}
  \, \sum_{\alpha\in\widehat{G}}\, \dim\alpha \ \chi_\alpha\big(\varphi(\mathcal{U}_w)\big)\, \exp\bigg(-e^2\,a_w\,\Big(\frac{ C_2(\alpha)}2+\upsilon_2\Big)\bigg)\\
& \hspace{2cm} =\e^{-\upsilon_1} \, \sum_{\alpha\in\widehat{G}}\,
  \dim\alpha \
  \chi_{\varphi^*\alpha}(\mathcal{U}_w)\, \exp\bigg(-e^2\,a_w\,\Big(
  \frac{C_2(\alpha)}2+\upsilon_2 \Big)\bigg) \ .
\end{align}  
Now notice that the dimensions of the representations $\alpha$ and $\varphi^* \alpha$ are the same.
The value of the quadratic Casimir operator $C_2$ is also the same 
in both representations, as discussed at the end of Section~\ref{Sec:Symmetry}.
Combining everything we get\footnote{This is also demonstrated
  in~\cite[Lemma~5.14]{Runkel:2018aqft} from an algebraic perspective,
  where it is shown that the outer automorphism $\varphi$ induces an
  isomorphism of commutative regularised Frobenius algebras.}
\begin{align}
{\mit\Gamma}\big(\mathcal{U}_w,e^2\,a_w\big) \xrightarrow{ \ \varphi \ }
  \e^{-\upsilon_1}\, \sum_{\alpha\in\widehat{G}}\, \dim
  \varphi^*\alpha \ \chi_{\varphi^*\alpha}(\mathcal{U}_w)\,
  \exp\bigg(-e^2\,a_w\,\Big( \frac{C_2(\varphi^*\alpha)}2+\upsilon_2
  \Big)\bigg)= {\mit\Gamma}\big(\mathcal{U}_w,e^2\,a_w \big)
\end{align} 
where in the last step we used the fact that $\varphi^*$ is a bijection on the set 
of isomorphism classes of irreducible representations. Since the lattice 
regularisation of the quantum gauge theory agrees with the continuum limit, this shows that $\varphi$
induces an actual symmetry at the quantum level. 

Let us now explicitly compute the partition function for a defect corresponding to 
the symmetry in terms of the combinatorial data of the discretisation
of $\Sigma$. The cellular description provided in this section is dual
to that of the triangulation used to define a defect network in
Section~\ref{Sec:Defect}; here defects correspond to turning edges of
the cell decomposition into symmetry twist branch cuts on $\Sigma$. A defect or domain wall corresponding to a symmetry can be 
implemented by performing the path integral over field 
configurations which change by the symmetry when 
passing through the domain wall. 
In the lattice gauge theory approach this has a simple implementation:
When calculating the holonomy around a face we count an edge $\gamma$ labeled by
$g_\gamma\in G$ which passes through a defect corresponding to $\varphi$ as $g_\gamma$ to the 
left of the defect and as $\varphi(g_\gamma)$ to the right of the
defect. A straightforward calculation shows that contractible defects
do not change the value of the partition function, so in the following
we focus on defects which wrap around non-contractible cycles of the
surface $\Sigma$. 

As a warm up, let us begin by calculating the partition function on a
genus one surface, which is a torus $\Sigma_1=T^2$, with a single
non-contractible defect line labeled by $\varphi\in\Out(G)$. We pick a cell decomposition of
$T^2$ with two edges and the defect as illustrated in Figure~\ref{Fig:
  Decomposition}.
\begin{figure}[htb]
\small
\begin{center}
\includegraphics[scale=0.8]{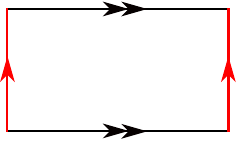}
\end{center}
\caption{\small Cell decomposition of $T^2$ into a rectangle with parallel edges
  identified. The defect is along the vertical edge.}\label{Fig: Decomposition}
\normalsize
\end{figure}    
For the path integral we have to specify the parallel transport along
two edges. The partition function can then be easily computed to give
\begin{align}
Z_{\textrm{\tiny YM}}\big(T^2,G,e^2\,a; \varphi\big)&= \int_{G\times
                                                      G} \, \dd g \
                                                      \dd h \ 
                                                \sum_{\alpha\in\widehat{G}}\,
                                                \dim\alpha \ 
                                                \chi_\alpha\big(\varphi(h)^{-1}\,g^{-1}\,h\,g\big)
                                                      \ \exp\bigg(-e^2\,a\,
                                                \Big(\frac{C_2(\alpha)}2+\upsilon_2\Big)\bigg)
                                                \\[4pt]
&= \int_{G} \, \dd h \ \sum_{\alpha\in\widehat{G}} \, \chi_\alpha\big({\varphi
  (h)}^{-1}\big) \ \chi_\alpha(h) \ \exp\bigg(-e^2\,a\, \Big(
  \frac{C_2(\alpha)}2+\upsilon_2\Big)\bigg) \\[4pt]
&= \int_{G} \, \dd h \ \sum_{\alpha\in\widehat{G}} \,
  \chi_{\varphi^*\alpha}(h^{-1}) \ \chi_\alpha(
  h) \ \exp\bigg(-e^2\,a\,
  \Big(\frac{C_2(\alpha)}2+\upsilon_2\Big)\bigg) \\[4pt]
&=  \sum_{\stackrel{\scriptstyle
  \alpha\in\widehat{G}}{\scriptstyle\alpha = \varphi^*\alpha}} \,
  \exp\bigg(-e^2\,a\, \Big(\frac{C_2(\alpha)}2 + \upsilon_2\Big)\bigg) \ ,
\end{align} 
where we used $ \chi(T^2)=0$ and $\chi_\alpha(1)=\dim\alpha$, together with the orthonormality and
fusion relations for the characters:
\begin{align}
\int_G \, \dd g \ \chi_\alpha (A\,g) \, \chi_\beta(g^{-1}\,B) &=
                                                                \delta_{\alpha,\beta}
                                                                \
                                                                \frac{1}{\dim
                                                                \alpha}
                                                                \
                                                                \chi_\alpha(A\,B)
                                                                \ , \\[4pt]
\int_G \, \dd g \ \chi_\alpha(A\,g\,B\,g^{-1}) &= \frac{1}{\dim
                                                 \alpha} \
                                                 \chi_\alpha(A)\,\chi_\alpha(B)
                                                 \ , \label{Eq: characters}
\end{align} 
with $A,B\in G$.
Setting $e=1$ for the renormalisation scheme with $\upsilon_2=0$, this reproduces the 
result of \cite{Runkel:2018aqft}. Representations
$\alpha\in\widehat{G}$ for which $\varphi^*\alpha=\alpha$ are called
\emph{fixed point representations} of the automorphism $\varphi$ in~\cite{Fuchs:1996twc}.

This calculation can be generalised to an 
arbitrary connected Riemann surface $\Sigma_p$ of genus $p>1$ and area $a$ containing $p$ defects, as illustrated
in Figure~\ref{Fig:Decomposition2}, labeled by group outer automorphisms $\varphi_1, \dots,
\varphi_p\in\Out(G)$.
\begin{figure}[htb]
\small
\begin{center}
\input{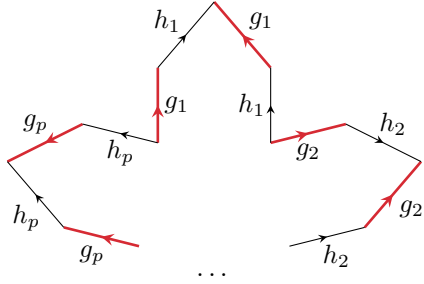}
\end{center}
\caption{\small Cell decomposition of $\Sigma_p$ into a $4p$-gon with
  edges of the same label
  identified. The defects are labeled by $h_i$.}\label{Fig:Decomposition2}
\normalsize
\end{figure}
Using the integral formulas \eqref{Eq: characters}, we then compute
\begin{align}
&Z_{\textrm{\tiny YM}}\big(\Sigma_p, G, e^2\,a ;  \varphi_1 , \dots ,
  \varphi_p\big) \label{eq:Sigmapcalc} \\ &\hspace{2cm} = \e^{\upsilon_1\,(2p-2)}\,
                     \int_{G^{\times 2p}} \ \prod_{j=1}^p \, \dd g_j
  \ \dd h_j \ 
                     \sum_{\alpha\in\widehat{G}}\, \dim\alpha \ 
                     \chi_\alpha\Big(\prod_{i=1}^p \,
                     g_i\,h_i\,\varphi_i(g_i^{-1})\,h_i^{-1}\Big)
                                            \nonumber \\
& \hspace{10cm} \times \ 
                     \exp\bigg(-e^2\,a\,\Big(\frac{C_2(\alpha)}2+\upsilon_2\Big)
                     \bigg) \nonumber \\[4pt] 
&\hspace{2cm} = \e^{\upsilon_1\,(2p-2)}\,\int_{G^{\times(2p-1)}} \, \dd g_p \
  \prod_{j=1}^{p-1}\, \dd g_j \ \dd h_j \ 
  \sum_{\alpha\in\widehat{G}} \, \chi_\alpha\Big( g_p \, \prod_{i=1}^{p-1}\,
  g_i\,h_i\,\varphi_i(g_i^{-1})\,h_i^{-1}\Big) \
  \chi_{\varphi_p^*\alpha}(g_p^{-1}) \nonumber \\ 
&\hspace{10cm} \times \ \exp\bigg(-e^2\,a\,
  \Big(\frac{C_2(\alpha)}2+\upsilon_2\Big)\bigg) \\[4pt]
&\hspace{2cm} =  \e^{\upsilon_1\,(2p-2)}\,\int_{G^{\times(2p-2)}}
  \ \prod_{j=1}^{p-1}\, \dd g_j \ \dd h_j \ 
  \sum_{\alpha\in\widehat{G}} \, \frac{\delta_{\alpha, \varphi_p^*
  \alpha}}{\dim \alpha} \ \chi_\alpha\Big(\prod_{i=1}^{p-1}\,
  g_i\,h_i\,\varphi_i(g_i^{-1})\, h_i^{-1}\Big) \nonumber \\ 
&\hspace{10cm} \times \ 
  \exp\bigg(-e^2\,a\,\Big(\frac{C_2(\alpha)}2+\upsilon_2\Big)\bigg) \nonumber
\end{align}
where we used $ \chi(\Sigma_p)=2-2p$. Proceeding inductively in this
way then
finally gives
\begin{align}\label{eq:ZYMgenusp}
&Z_{\textrm{\tiny YM}}\big(\Sigma_p, G, e^2\,a ; \varphi_1 , \dots ,
  \varphi_p\big) \\
&\hspace{4cm} =
  \e^{\upsilon_1\,(2p-2)}\,
  \sum_{\stackrel{\scriptstyle\alpha\in\widehat{G}}{\scriptstyle
  \alpha=\varphi_1^*\alpha= \cdots = \varphi_p^*\alpha}} \, (\dim
  \alpha)^{2-2p} \ \exp\bigg(-e^2\,a\,\Big( \frac{C_2(\alpha)}2
  + \upsilon_2\Big) \bigg) \nonumber
\end{align}  
for the Yang-Mills partition function on an oriented Riemann surface $\Sigma_p$
of genus $p$ with defects labeled by
$\varphi_1,\dots,\varphi_p\in\Out(G)$ around non-contractible
cycles of $\Sigma_p$. When all defects are trivial, $\varphi_i=\id_G$ for
$i=1,\dots,p$, this combinatorial expression is just the usual
Migdal-Rusakov heat kernel expansion for the partition function of Yang-Mills theory
on
$\Sigma_p$~\cite{Migdal:1975re,Rusakov:1990rs,Witten:1991gt,Blau:1992ym,Cordes:1995ym}. In
general it agrees with the computation
of~\cite[Proposition~5.17]{Runkel:2018aqft} for the particular defect
configuration at hand.

In at least simple cases, the partition function \eqref{eq:ZYMgenusp} can be computed
explicitly using the combinatorics of Dynkin diagrams. For this,
recall that outer automorphisms of a semi-simple Lie algebra $\mathfrak{g}$  
are in one-to-one correspondence with automorphisms of the underlying
Dynkin diagram. Let $(C_{i,j})_{i,j \in I_r}$ be the
Cartan matrix encoded by the corresponding Dynkin diagram, where
$I_r=\{1,\dots , r \}$ and $r$
is the rank of $\mathfrak{g}$. An automorphism of the Dynkin diagram is a bijective
map $\varphi \colon I_r \longrightarrow I_r$ which preserves the
entries of the Cartan matrix:
$C_{i,j}=C_{\varphi(i),\varphi(j)}$ for all $i,j\in I_r$. Associated to
the Dynkin diagram is a Cartan-Weyl basis of Chevalley generators $\{
H_i, E^\pm_i \}_{i\in I_r}$ of $\mathfrak{g}$ in which $\varphi$ induces the outer automorphism 
\begin{align}
\varphi \colon \mathfrak{g} \longrightarrow \mathfrak{g} \ , \quad
\big(H_i,E^{\pm}_i\big) \xrightarrow{ \ \varphi \ }
  \big(H_{\varphi(i)} ,E^\pm_{\varphi(i)} \big)
\end{align}
of $\mathfrak{g}$.
This in turn induces an isomorphism between the group of symmetries of the
underlying Dynkin diagram and the group of 
outer automorphisms of the Lie algebra $\mathfrak{g}$, see for
example~\cite{Fuchs:1996twc}. Let $\mathfrak{h}\subset \mathfrak{g}$
be the Cartan subalgebra spanned by $H_i$. Then the automorphism
$\varphi$ induces an action on the weight space $\mathfrak{h}^*$ given
by pullback
\begin{align}
\varphi^*:\mathfrak{h}^* \longrightarrow \mathfrak{h}^* \ , \quad
\lambda(x) \xrightarrow{ \ \varphi^* \ } (\varphi^*\lambda) (x) \coloneqq  \lambda\big(\varphi^{-1}(x)\big) \ .
\end{align}
A weight vector $\lambda$ is \emph{symmetric} if
$\varphi^*\lambda=\lambda$. If $\alpha$ is 
an irreducible representation of $\mathfrak{g}$ with highest weight 
vector $v$ of weight $\lambda_\alpha$, then $v$ is also a highest weight 
vector for the representation $\varphi^*\alpha$ with weight
$\varphi^*\lambda_\alpha$~\cite[Section 4]{Fuchs:1996twc}. For the
calculation of the partition function \eqref{eq:ZYMgenusp}, we thus have to restrict to representations with symmetric highest weight. 

The highest weights corresponding to irreducible representations of $\mathfrak{g}$ can be uniquely expressed as 
\begin{align}
\lambda = \sum_{i=1}^r\, n_i \, \omega_i \ ,
\end{align} 
where $\omega_i$ are the fundamental weights and $n_i \in \N_0$ for 
$i=1,\dots,r$. The non-negative integers $n_i$ are the \emph{Dynkin labels} of the 
corresponding representation. The action of $\varphi$ on the fundamental weights is given by $\varphi^*\omega_i = \omega_{\varphi(i)}$. Hence a representation $\alpha$ is symmetric if and only if the corresponding Dynkin labels are invariant under the transformation
\begin{align}
[n_1,\dots,n_r] \xrightarrow{ \ \varphi \ } [n_{\varphi^{-1}(1)},\dots, n_{\varphi^{-1}(r)}] \ .
\end{align}
The Dynkin labels can be used to concretely calculate the partition
function \eqref{eq:ZYMgenusp} for a genus $p$ surface $\Sigma_p$
containing defects: The dimension and quadratic Casimir invariant of
an irreducible unitary representation with highest weight $\lambda$
are given by
\begin{align}
\dim\lambda = \prod_{\alpha\in \mathfrak{R}_+} \,
  \frac{(\lambda+\rho,\alpha)_{\mathfrak{g}^*}}{(\rho,\alpha)_{\mathfrak{g}^*}}
  \qquad \mbox{and} \qquad C_2(\lambda) = (\lambda+2\rho,\lambda)_{\mathfrak{g}^*} \ ,
\end{align}
where $\mathfrak{R}_+\subset\mathfrak{h}^*$ is the system of positive roots of the Lie
algebra $\mathfrak{g}$, $\rho=\frac12\,\sum_{\alpha\in\mathfrak{R}_+}\,\alpha$ is the Weyl vector, and the
invariant bilinear form $(\,\cdot\,,\,\cdot\,)_{\mathfrak{g}^*}$ on
$\mathfrak{g}^*$ is
induced by the Killing form of $\mathfrak{g}$.

\begin{example}
The Lie group $G=SU(3)$ has rank $r=2$ and hence every irreducible representation can 
be labeled by a pair of non-negative integers $[n,m]$. $SU(3)$ admits
only one non-trivial
outer automorphism $\varphi$ corresponding to complex conjugation, which 
acts on the Dynkin labels by interchanging $n$ and $m$. Hence
$\Out\big(SU(3)\big)=\Z_2$ and
symmetric representations are real representations which
are of the form $[n,n]$. The dimension of the representation $[n,n]$ is 
$\dim [n,n]=(n+1)^3$, and the value of its quadratic Casimir invariant
is $C_2\big([n,n]\big)=n\,(n+2)$. Consider the defect network from
Figure~\ref{Fig:Decomposition2} with at least one non-trivial defect label
$\varphi$. Then the partition function \eqref{eq:ZYMgenusp} reads as
\begin{align}\label{Eq: SU(3)}
Z_{\textrm{\tiny YM}}\big(\Sigma_p,SU(3),e^2\,a;\varphi\big)=
  \e^{\upsilon_1\,(2p-2)} \, \sum_{n=0}^\infty \, (n+1)^{6-6p} \
  \exp\bigg(-e^2\,a\,\Big( \frac{n\,(n+2)}2 + \upsilon_2\Big) \bigg) \ .
\end{align}    
\end{example}

\subsection{The moduli space of flat twisted bundles}\label{Sec:Zero}

Let $D$ be an $\Out(G)$-bundle on a surface $\Sigma$. 
We denote by $\mathscr{M}_G^D(\Sigma)$ the moduli space of flat $D$-twisted $G$-bundles
on $\Sigma$, or in other words pairs $(P,A)\in \Bun^\nabla_{G\downarrow D}(\Sigma)$ with
  $F_A=0$, up to gauge transformations. 
Up to equivalence we can describe $D$ by a group homomorphism on the
fundamental group of $\Sigma$:
$\kappa_D\colon \pi_1(\Sigma)\longrightarrow \Out(G)$. 
A flat $D$-twisted $G$-bundle $P$ on $\Sigma$ can then be described by a group 
homomorphism $\phi'_P\colon \pi_1(\Sigma)\longrightarrow G\rtimes \Out(G)$
which lifts the group homomorphism~$\kappa_D$ in the sense that the
diagram
\[
\begin{tikzcd}
 & & G\rtimes\Out(G) \ar[dd] \\
 & & \\
 \pi_1(\Sigma) \ar[uurr, "\phi_P'"] \ar[rr, swap,"\kappa_D"] &  & \Out(G)
\end{tikzcd}
\]
commutes, where the vertical arrow is the projection to the second
factor.

Equivalently, this can be described by a map $\phi_P\colon \pi_1(\Sigma) \longrightarrow G$ satisfying
\begin{align}
\phi_P(\gamma_1 * \gamma_2)=\phi_P(\gamma_1)\,
  \kappa_D(\gamma_1)\big(\phi_P(\gamma_2)\big) \ , 
\end{align}
where $\gamma_1\ast\gamma_2$ denotes the concatenation of paths on $\Sigma$
between representatives of the corresponding homotopy classes. 
Let $\Hom_{\kappa_D}\big(\pi_1(\Sigma),G\big)$ denote the space of all such
twisted group homomorphisms; for any $\phi\in \Hom_{\kappa_D}\big(\pi_1(\Sigma),G\big)$ and any
homotopy class of 
paths $[\gamma]\in\pi_1(\Sigma)$, $\phi(\gamma)$ is the holonomy of a
flat $D$-twisted $G$-connection along $\gamma$. 
Gauge transformations correspond to the action of the Lie group $G$ on this space via the twisted conjugation 
\begin{align}
g\cdot\phi \colon \pi_1(\Sigma) \longrightarrow G \ , \quad
\gamma    \longmapsto (g\cdot\phi)(\gamma)=g\, \phi(\gamma) \,
  \kappa_D(\gamma)\big(g^{-1} \big) \ . 
\end{align}
The moduli space $\mathscr{M}_G^D(\Sigma)$ can be identified with the 
quotient $\Hom_{\kappa_D}\big(\pi_1(\Sigma),G\big)/ G$ by this $G$-action. 

In the local triangulation
description of defect networks from Section~\ref{Sec:Defect}, a flat $D$-twisted $G$-bundle
 is the same as a
$D$-twisted $G$-local system on $\Sigma$, as defined for example
in~\cite{Labourie2013}; one may also characterise it as a groupoid homomorphism from
the fundamental groupoid of $\Sigma$ to the classifying groupoid
of $G\rtimes\Out(G)$-bundles.
It is possible to generalise this description to surfaces
$\Sigma$ with
boundary circles by using a subgroupoid of the fundamental groupoid,
and hence to describe moduli spaces of flat $D$-twisted
$G$-connections on $\Sigma$ with holonomies on the boundary components in
prescribed twisted conjugacy classes of $G$, see~\cite{Meinrenken:Convexity} 
for further details. 

To relate this moduli space to the quantum gauge theory defined in 
Section~\ref{Sec:Defect}, we note that the weak-coupling limit $e\longrightarrow 0$ of 
the Yang-Mills action functional \eqref{eq:YMaction} is either $0$ or it diverges to $+\infty$, and hence the path integral \eqref{Eq: Path integral definition defect}
localises onto gauge field configurations with vanishing action
functional, or equivalently with vanishing curvature $F_A=0$; these
are precisely the flat 
twisted bundles. Hence the path integral formally reduces to an integral
over $\mathscr{M}_G^D(\Sigma)$. 
Since the integration measure
$\mathscr{D}(P,A)$ is formally induced by the
infinite-dimensional symplectic structure \eqref{eq:SymplecticForm}, it is natural to conjecture that 
the partition function \eqref{Eq: Path integral definition defect} computes the symplectic volume of 
$\mathscr{M}_G^D(\Sigma)$ in the weak-coupling limit, where the symplectic
two-form on $\mathscr{M}_G^D(\Sigma)$ is inherited from \eqref{eq:SymplecticForm}. This argument is
completely analogous to that given in the case of ordinary Yang-Mills
theory in~\cite{Witten:1991gt}.  

To describe the weak-coupling limit more
precisely as a topological field theory, it is useful 
to consider an equivalent formulation of the quantum Yang-Mills theory in the presence of defects.
For this, recall that a $D$-twisted $G$-bundle on $\Sigma$ can be described by a $G\rtimes \Out(G)$-bundle
$P$. The curvature $F_A$ of the connection on the bundle $P$ is a
two-form on $\Sigma$ with values in the associated
$\mathfrak{g}$-bundle $\Ad(P)$.
We introduce an auxiliary scalar field $\phi$ on $\Sigma$
with values in $\Ad(P)$, and consider the action functional
\begin{align}
S(P,A,\phi)= - \iu \, \int_{\Sigma}\, \Tr_{\mathfrak{g}}\big(\phi \, F_A\big)
  -\frac{e^2}{2} \, \int_{\Sigma}\, \dd\mu \ \Tr_{\mathfrak{g}}\big(\phi^2\big) 
  \ .
\end{align} 
The field $\phi$ can only be defined after the $D$-twisted bundle
$(P,A)$ is fixed, and the corresponding path integral
\begin{align}\label{eq:BFtheory}
\int_{\Bun^\nabla_{G\downarrow D}(\Sigma)} \, \mathscr{D}(P,A) \
  \int_{\Omega^0(\Sigma;\Ad(P))} \, \mathscr{D}\phi \ 
  \exp\big(-S(P,A,\phi)\big) 
\end{align}
is taken over all $D$-twisted bundles with connections and 
$\phi\in\Omega^0(\Sigma;\Ad(P))$, where the measure on the space
$\Omega^0(\Sigma;\Ad(P)$ is induced by the metric on $\Ad(P)$ given by
\begin{align}
\|\phi\|^2 := -\frac1{4\pi^2}\, \int_\Sigma\,
\dd\mu \ \Tr_{\mathfrak{g}}\big(\phi^2\big) \ .
\end{align}
Performing the Gaussian path integral over $\phi$ (or equivalently
eliminating $\phi$ by its Euler-Lagrange equation) 
shows that the quantum field theory defined by \eqref{eq:BFtheory} is
equivalent to the quantum field theory defined by \eqref{Eq: Path
  integral definition defect} for two-dimensional Yang-Mills theory in
the presence of symmetry defects. The path integral
\eqref{eq:BFtheory} is subject to the same two-parameter
renormalisation ambiguity \eqref{eq:ambiguity}, which multiplies it by
the factor
$\exp(-\Delta S)$.

The advantage of this reformulation is that it is straightforward now to
take the $e\longrightarrow 0$ limit, 
which is described by the topological field theory with action
functional 
\begin{align}
S_{0}(P,A,\phi)=  - \iu \, \int_{\Sigma}\, \Tr_{\mathfrak{g}}\big(\phi \, F_A \big) \ .
\end{align}
At $e=0$, the invariance under the group of area-preserving
diffeomorphisms is promoted to full diffeomorphism invariance, and the
ambiguity \eqref{eq:ambiguity} is reduced to a one-parameter ambiguity
depending only on the topology of the surface $\Sigma$.
This looks similar to the quantum field theory describing the weak-coupling limit of 
Yang-Mills theory without defects~\cite{Witten:1991gt,Witten:1992xu}.
The difference is that here $\phi$ takes values in a different bundle and that the path
integral is taken over twisted bundles rather than ordinary
bundles. Integrating over $\phi$ in \eqref{eq:BFtheory} at $e=0$
produces a formal delta-functional $\delta(F_A)$, and the
remaining path integral over $\Bun^\nabla_{G\downarrow D}(\Sigma)$ therefore localises on the locus $F_A=0$, which by
definition is the moduli space $\mathscr{M}_G^D(\Sigma)$ of flat
$D$-twisted $G$-bundles on~$\Sigma$. In the usual untwisted
  case~\cite{Witten:1991gt}, the argument showing that the resulting path integral
  measure induces the correct symplectic volume form on the moduli space $\mathscr{M}_G(\Sigma)\simeq\Hom\big(\pi_1(\Sigma),G\big)/ G $ of flat
  $G$-connections on $\Sigma$ uses a careful application of Faddeev-Popov gauge
  fixing and the triviality of analytic torsion on oriented 
  surfaces, together with a judicious choice of $\upsilon_1$. It should be possible to extend these arguments to the
  twisted case.

Putting everything together, we conjecture that the
symplectic volume of $\mathscr{M}_G^D(\Sigma)$ can be given a gauge theory
interpretation via the formula
\begin{align}
{\rm Vol}\big(\mathscr{M}_G^D(\Sigma)\big)= \e^{-\upsilon_1\,  \chi(\Sigma)} \
  \lim\displaylimits_{e\rightarrow 0} \, Z_{\textrm{\tiny YM}}\big(\Sigma,G,e^2\,a;D\big) \ .
\end{align} 
Since the undetermined parameter $\upsilon_1\in\R$ depends only on $G$ and
the renormalisation scheme, but not on $\Sigma$, the ratio
\begin{align}
\frac{{\rm Vol}\big(\mathscr{M}_G^D(\Sigma)\big)}{{\rm Vol}\big(\mathscr{M}_G(\Sigma)\big)} =
  \lim\displaylimits_{e\rightarrow 0} \, \frac{Z_{\textrm{\tiny
  YM}}\big(\Sigma,G,e^2\,a;D\big)}{Z_{\textrm{\tiny YM}}\big(\Sigma,G,e^2\,a\big)}
\end{align}
is independent of the choice of the renormalisation scheme. This ratio can thus
be computed explicitly using the lattice regularisation of
Section~\ref{Sec:Lattice}, and used to make concrete predictions for
the symplectic volume ${\rm Vol}\big(\mathscr{M}_G^D(\Sigma)\big)$; in
the lattice formulation, a connection
is flat if $\mathcal{U}_w=1$ for every face $w\in\Sigma^{\textrm{\tiny(2)}}$.

\begin{example}
Let us look again at the simplest non-trivial example of gauge group
$G=SU(3)$. From \eqref{Eq: SU(3)} we deduce that the volume of $\mathscr{M}_{SU(3)}^D(\Sigma_p)$
is independent of the choice of non-trivial $\Z_2$-bundle
$D\longrightarrow \Sigma_p$. 
For genus $p\geq 2 $ the weak-coupling limit $e \longrightarrow 0$ exists and
the series sums to give
the value of the Riemann zeta-function $\zeta(6p-6)$. It follows that
the symplectic volume in the presence of defects is
\begin{align}
{\rm Vol}\big(\mathscr{M}_{SU(3)}^D(\Sigma_p)\big)=\lim\displaylimits_{e \rightarrow 0} \,
  Z_{\textrm{\tiny YM}}\big(\Sigma_p,SU(3),e^2\,a;\varphi\big) = \e^{\upsilon_1\,(2p-2)} \, \zeta(6p-6) \ .
\end{align} 
The undetermined parameter $\upsilon_1$ can in principal be determined
from the results for untwisted bundles;
in~\cite{Witten:1991gt,Witten:1992xu} the constant $\upsilon_1$ is
evaluated by a direct computation of the Reidemeister
torsion. However, in the present case we are not able to determine
{\it a priori} the value
of $\upsilon_1$, so we cannot make a more
concrete prediction for the symplectic volume at this stage. 
\end{example}

\subsection{The orbifold Yang-Mills theory}\label{Sec:Orbifold}

Given a discrete symmetry of a quantum field theory one can try to
gauge the symmetry, or in other words construct a corresponding 
orbifold theory; the orbifold field theory is constructed by
taking the quotient by the symmetry group and projecting the Hilbert space onto the invariant states.
In this paper we focus on the defect approach to orbifolds~\cite{FFSR:2009orb,Brunner:2013orb}
and show that the orbifold theory corresponding to the symmetry introduced in Section \ref{Sec:Symmetry} is the Yang-Mills theory 
with gauge group $G\rtimes \Out(G)$. We can further naturally
twist the orbifold theory by a two-cocycle $c\in H^2(\Out(G);U(1))$
representing the inclusion of discrete torsion. The resulting orbifold
theory will then be a two-dimensional Yang-Mills theory based on the structure group 
$G\rtimes \Out(G)$ with a topological Dijkgraaf-Witten term~\cite{DijkgraafWitten} for the
finite group $\Out(G)$ added to the Yang-Mills action functional
\eqref{eq:YMaction}; this corresponds to coupling the Yang-Mills
theory to a two-dimensional symmetry protected topological phase,
which is specified by the two-cocycle $c$ and protected by the $\Out(G)$-symmetry. 

Let ${\tt D}_\varphi$ denote the defect corresponding to an outer automorphism 
$\varphi\in \Out(G)$. We construct the orbifold defect as the superposition
\begin{align}
{\tt P}_G =  \sum_{\varphi\in \Out(G)} \, {\tt D}_\varphi \ ,
\end{align}
corresponding to a superposition of $\Out(G)$-bundles over $\Sigma$.
The partition function of the orbifold theory on a Riemann surface $\Sigma$ can
be constructed by picking a triangulation of $\Sigma$ and computing the partition
function of the original Yang-Mills theory in the presence of a defect network
where every edge of the triangulation is labeled with ${\tt P}_G$. The intersections
need to be labeled by `junction fields' which introduce an 
appropriate normalisation; we will explain this in more detail in 
Section~\ref{Sec:Funct}. 
In practice this reduces to a sum over all consistent defect labels of the triangulation
with a normalisation factor $\frac1{|\Out(G)|^V}$, where $V$ is the
number of vertices of the triangulation.
Recall from Section~\ref{Sec:Defect} that, for a fixed defect configuration, the
path integral is taken over a subspace of $\Bun_{G\rtimes \Out(G)}^\nabla(\Sigma)$.
The sum over all labels for defect lines reduces to a sum over all possible 
$\Out(G)$-bundles. As a consequence, the partition function of the orbifold theory 
can be interpreted as an integral over the entire space 
$\Bun_{G\rtimes \Out(G)}^\nabla(\Sigma)$. Dividing the result by $|\Out(G)|^V$
takes care of the fact that in Section~\ref{Sec:Defect} we only divided out 
$G$-gauge transformations; the additional normalisation correctly takes 
care of the discrete part. 
This indicates that that the partition function
of the orbifold theory agrees with the partition function of
Yang-Mills theory on $\Sigma$ with gauge group $G\rtimes \Out(G)$.
By adding a two-dimensional Dijkgraaf-Witten term for $\Out(G)$
into the sum we can also construct a twisted version of this orbifold
Yang-Mills theory with coupling to an $\Out(G)$-symmetry protected topological
phase.   

We now turn our attention to the state space of the orbifold theory, which can be constructed 
by first adding twisted sectors to the original state space to get
\begin{align}\label{eq:Hprime}
\Hc_G'= \bigoplus_{\varphi\in \Out(G)} \, Z_{\textrm{\tiny
  YM}}\big(S^1,G;\varphi\big) \ ,
\end{align}
where $Z_{\textrm{\tiny YM}}\big(S^1,G;\varphi\big)$ is the state space on a
Cauchy circle $S^1$ in $\Sigma$ in the presence of a point defect labeled by
$\varphi$. This is the Hilbert space of gauge-invariant functions on the space of twisted bundles 
over $S^1$. The only gauge-invariant quantity that can be constructed
from a twisted bundle on $S^1$ is its holonomy
$\mathcal{U}$, which transforms under a gauge transformation
corresponding to $g\in G$ as $\mathcal{U}\longmapsto g\, \mathcal{U}\,
\varphi(g^{-1})$. 
This shows that the state space for each twisted sector
$\varphi\in\Out(G)$ is given by
\begin{align}
Z_{\textrm{\tiny YM}}\big(S^1,G;\varphi\big)= \big\lbrace f\in L^2(G)
  \, \big| \, f(g)= f\big(h \, g \,
  \varphi(h^{-1})\big) \ \mbox{for all}\ g,h \in G \big\rbrace \ .
\end{align}

The Hilbert space $Z_{\textrm{\tiny YM}}\big(S^1,G;\varphi\big)$ has a
natural basis given by `twining characters', which we describe explicitly 
following~\cite{Fuchs:1996twc} for the special case of semi-simple Lie
groups, see also~\cite{Zerouali:2018tw}. Let $\omega \colon
G\longrightarrow G$ be an outer automorphism of the Lie group $G$ constructed from an
automorphism of the corresponding Dynkin diagram, and let $\pi_\alpha \colon \mathfrak{g}
\longrightarrow \End(V_\alpha)$ be a corresponding fixed point unitary irreducible highest weight
representation: $\omega^*\alpha = \alpha$. 
Then by Schur's lemma there exists a unitary automorphism $T_\alpha^\omega \colon
V_\alpha \longrightarrow V_\alpha$ such that the diagram
\begin{equation}
\begin{tikzcd}
V_\alpha \ar[dd, "T_\alpha^\omega",swap] \ar[rr,"\pi_\alpha(\omega(v))"] & & V_\alpha \ar[dd, "T_\alpha^\omega"] \\ 
 & & \\
V_\alpha \ar[rr,"\pi_\alpha(v)",swap] & & V_\alpha
\end{tikzcd}
\end{equation} 
commutes for all $v\in V_\alpha$. Since $\omega$ only
permutes the generators of $\mathfrak{g}$, it preserves the highest
weight space. Requiring $T_\alpha^\omega$ to be the identity on the
highest weight space thus fixes it uniquely. 
Then $T_{\omega_2}^\alpha \, T_{\omega_1}^\alpha= T_{\omega_1\,
  \omega_2}^\alpha$ for any two automorphisms $\omega_1, \omega_2\in \Out(G)$.      

Exponentiating the representation we get a corresponding
representation $\pi_\alpha \colon G\longrightarrow \End(V_\alpha)$ of
the group $G$. The twining characters
$\chi_\alpha^\omega:G\longrightarrow \C$ can now be defined by 
\begin{align}
\chi^\omega_\alpha (g) \coloneqq \tr^{\phantom{\dag}}_{V_\alpha} \big(\pi_\alpha(g) \, T_\omega^\alpha\big) 
\end{align}
for $g\in G$.
They satisfy the twisted conjugation invariance 
\begin{align}
\chi_\alpha^\omega\big(h\,g\,\omega(h^{-1})\big) &= \tr^{\phantom{\dag}}_{V_\alpha}
                                                 \Big(\pi_\alpha\big(h\,g\,\omega(h^{-1})\big)
                                                 \,
                                                 T_\omega^\alpha\Big)
  \\[4pt]
&= \tr^{\phantom{\dag}}_{V_\alpha} \big(\pi_\alpha(h\,g) \, T_\omega^\alpha \, \pi_\alpha(h^{-1})
  \, {T_\omega^\alpha}^\dag \, T_\omega^\alpha\big) \\[4pt]
&= \tr^{\phantom{\dag}}_{V_\alpha} \big(\pi_\alpha(g) \, T_\omega^\alpha\big) \\[4pt]
&= \chi_\alpha^\omega(g) \ ,
\end{align}
for all $g,h\in G$.
Using the orthogonality of the matrix element functions it is easy to show that the 
twining characters span the Hilbert space $Z_{\textrm{\tiny YM}}\big(S^1,G;\omega\big)$ and satisfy a
generalisation of the orthonormality and fusion relations \eqref{Eq: characters} given by (see also~\cite{Zerouali:2018tw})
\begin{align}
\int_G\,
  \dd g \ \chi_\alpha^\omega(A\,g)\,\chi_\beta^{\omega'}(g^{-1}\,B) &= \delta_{\alpha,\beta} \ \frac{1}{\dim \alpha} \
          \chi_\alpha^{\omega' \, \omega}(A\,B) \ , \\[4pt]
\int_G\, \dd g \ \chi_\alpha^\omega(A\,g\,B\,g^{-1}) & =\frac{1}{\dim \alpha} \
                                              \chi_\alpha(A) \, \chi_\alpha^\omega(B)
                                              \ , \label{Eq: twining Characters}
\end{align}
for $A,B\in G$.

We now note that \eqref{eq:Hprime}
is not the Hilbert space of the orbifold theory, because there is
a natural $\Out (G)$-action on $\Hc'_G$ and the Hilbert space
of the orbifold theory is the subspace of invariants. An outer automorphism $\omega\in \Out(G)$
maps $f\in Z_{\textrm{\tiny YM}}\big(S^1,G;\varphi\big)$ to
$\omega\cdot f \in Z_{\textrm{\tiny YM}}\big(S^1,G;\omega\, \varphi\,
\omega^{-1}\big)$ via the linear map 
corresponding to the defect network illustrated in
Figure~\ref{Fig:Defect-action}.
\begin{figure}[htb]
\small
\begin{center}
\input{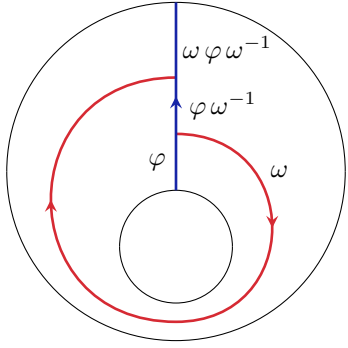}
\end{center}
\caption{\small The action of $\omega\in\Out(G)$ on a defect labeled
  by $\varphi$.}
\label{Fig:Defect-action}
\normalsize
\end{figure} 
A straightforward computation using \eqref{Eq: twining Characters} and 
the lattice regularisation
shows that the action is given by $\omega\cdot f \coloneqq
f\circ\omega^{-1}$.\footnote{See 
  Proposition~\ref{prop:orbifold-projector} below for a rigorous proof of this
  statement.}
To describe the space of invariants, we pick a representative $\sf C$
for every
conjugacy class of $\Out(G)$. For an automorphism $\varphi\in \Out(G)$, we denote by $\Com(\varphi)$
the commutant of $\varphi$ in $\Out(G)$, or in other words the subgroup of $\Out(G)$ commuting with $\varphi$;
the action of $\Com(\varphi)$ on $\Hc'_G$ preserves $Z_{\textrm{\tiny
    YM}}\big(S^1,G;\varphi \big)$.
A description of the state space of the orbifold theory which depends on the choice
of conjugacy class $\sf C$ is then given by
\begin{align}
\mathcal{H}_{G,\sf C} = \bigoplus_{\varphi\in {\sf C}} \, Z_{\textrm{\tiny YM}}\big(S^1,G;\varphi\big)^{\Com(\varphi)} \ ,
\end{align}
where we denote by $Z_{\textrm{\tiny
    YM}}(S^1,G;\varphi)^{\Com(\varphi)}$ the subspace of
invariants in $Z_{\textrm{\tiny
    YM}}(S^1,G;\varphi)$
with respect to the $\Com(\varphi)$-action. 

Now we argue that the Hilbert space $\Hc_{G,\sf C}$ is the space
$\Cl\big(G\rtimes \Out(G)\big)$ of class functions on the group
$G\rtimes\Out(G)$, confirming that the
orbifold theory is indeed the Yang-Mills theory based on $G\rtimes \Out(G)$.
For this, we note that a conjugation-invariant function on
$G\rtimes\Out(G)$ is completely determined
by its values on elements of $G\times {\sf C}\subset G\rtimes \Out(G)$. 
Hence we can describe any function $f\in \Cl\big(G\rtimes \Out(G)\big)$ by a family of functions
$f_\varphi\colon G\longrightarrow \C$ labeled by the elements
$\varphi\in\sf C$. The value of a function $f_\varphi(g)$ on $g\in G$ transforms under conjugation with respect 
to elements of the form $(h, 1)\in G\rtimes\Out(G)$ as 
$f_\varphi(g) \longmapsto f_\varphi\big(h\,g\,\varphi(h^{-1}) \big)$. 
This shows that $f_\varphi\in Z_{\textrm{\tiny
    YM}}\big(S^1,G;\varphi\big)$. The function $f_\varphi$ 
is further required to be invariant under conjugation by elements of the form
$(1, \omega)\in G\rtimes\Out(G)$ with $\omega\in \Com(\varphi)$, which induces the transformation
$f_\varphi\longmapsto f_\varphi\circ\omega$. Hence
$f_\varphi\in Z_{\textrm{\tiny
    YM}}(S^1,G;\varphi)^{\Com(\varphi)}$, as required.
See Proposition~\ref{prop:orbifold-projector} below for an explicit
description of the inverse map $Z_{\textrm{\tiny
    YM}}(S^1,G;\varphi)^{\Com(\varphi)}\longrightarrow
\Cl(G\rtimes\Out(G))$.   

Generally there are obstructions to the construction of an orbifold
theory for a quantum field theory with a discrete symmetry. However, all of
these obstructions vanish in the case considered in this paper, since we can construct the orbifold theory
explicitly. This is reminescent of the situation for finite
gauge groups where classical symmetries can be described via group extensions
by the symmetry group, and the field theory can be gauged if the action
functional of the 
original gauge theory can be lifted to the extension~\cite{Kapustin:Symmetries,MSGauge}.
In the continuous case considered here the extension is given by the
semi-direct product
\begin{align}
1\longrightarrow G \longrightarrow G\rtimes \Out(G)\longrightarrow \Out(G) \longrightarrow 1 \ ,
\end{align}  
and a lift of the action functional is provided by the action
functional of the Yang-Mills theory
with gauge group $G\rtimes \Out(G)$. 

\subsection{The reverse orbifold Yang-Mills theory}\label{sec:reverseorb}

It is possible to return back to the original Yang-Mills theory via a generalised orbifold construction.
We briefly sketch the construction here and refer to
Section~\ref{Sec:Funct} for further mathematical details.
We denote by $W_c$ a set of representatives for the 
isomorphism classes $c\in\widehat{\Out(G)}$ of irreducible representations of $\Out(G)$.
Via pullback by the homomorphism $G\rtimes \Out(G)\longrightarrow \Out(G)$, these induce representations of $G\rtimes \Out(G)$
and hence Wilson line defects in the Yang-Mills theory with gauge group
$G\rtimes \Out(G)$, which we denote again by $W_c$. These defects 
are invertible only if $\Out(G)$ is an abelian group, so that $W_c$
are all one-dimensional vector spaces. 
But they are always topological defects and, in particular, they are trivial for contractible 
loops, since the $\Out(G)$ part of the holonomy of any $G\rtimes \Out(G)$-connection around a contractible loop is trivial. 

To reverse the orbifold construction we use the defect 
\begin{align}\label{eq:reversedefect}
\Ac_G=\bigoplus_{c\in\widehat{\Out(G)}} \, W_c^{\oplus w_c} \simeq {L}^2\big(\Out(G)\big) \ ,
\end{align}
where $w_c=\dim W_c$.
Only when $\Out(G)$ is an abelian group does this defect come from a
symmetry, as in that case \eqref{eq:reversedefect}
decomposes into a direct sum of invertible defects. For non-abelian
groups $\Out(G)$ we need the generalised orbifold construction 
of \cite{FFSR:2009orb,Brunner:2013orb} to go backwards.
The reason why the choice of defect \eqref{eq:reversedefect} works is that the character of the regular representation ${L}^2(\Out(G))$ of $\Out(G)$ is given by
$\chi_{{L}^2(\Out(G))}(\kappa)= |\Out(G)| \, \delta_{\kappa,\id_G}$
for $\kappa\in\Out(G)$. Inserting a Wilson 
loop corresponding to $\Ac_G$ into the path integral for a Riemann surface $\Sigma$ localises the integration domain to $G\rtimes \Out(G)$-bundles with trivial $\Out(G)$ holonomy around the inserted loop.
If at least one Wilson loop for every generator of the fundamental 
group $\pi_1(\Sigma)$ labeled by $\Ac_G$ is inserted into the path integral,
then the $\Out(G)$ part of all bundles contributing to the partition function
is trivial and hence the partition function reduces to the partition
function \eqref{Eq: Partition function} of Yang-Mills theory
with gauge group $G$. We will prove this rigorously in Section~\ref{Sec:Funct}
using the orbifold completion of the topological defect bicategory of
two-dimensional Yang-Mills theories. 

We conclude by describing the reverse orbifold Yang-Mills theory in the lattice
regularisation of Section~\ref{Sec:Lattice}.
To compute the orbifold gauge theory we have to evaluate the 
partition function in the presence of a sufficiently dense
defect network labeled by \eqref{eq:reversedefect} with appropriate junction fields inserted. 
The junction fields correspond to the pointwise multiplication of functions
and the comultiplication
\begin{align*}
\Delta \colon {L}^2\big(\Out(G)\big) & \longrightarrow
                                       {L}^2\big(\Out(G)\big)\otimes
                                       {L}^2\big(\Out(G)\big) \simeq 
                                       {L}^2\big(\Out(G)\times
                                       \Out(G)\big) \ , \\
\big[\kappa\mapsto f(\kappa)\big] &\longmapsto
            \big[(\kappa_1,\kappa_2)\mapsto\delta_{\kappa_1,\kappa_2}\,
            f(\kappa_1) \big] \ .
\end{align*}  
These maps are homomorphisms of $\Out(G)$-representations and hence induce 
homomorphisms between the corresponding representations of $G\rtimes \Out(G)$. 
The associated junction fields are then given by
\begin{align}
\sum_{\kappa \in \Out(G)}\, \kappa \otimes \kappa \otimes \delta_{\kappa,\id_G} \ &\in \
 \C\big[\Out (G)\big]\otimes \C\big[\Out (G)\big] \otimes {L}^2\big(\Out
                                                                            (G)\big)
                                                                            \
  , \\[4pt] \label{Eq. Junction fields}
\sum_{\kappa \in \Out(G)}\, \kappa \otimes \delta_{\kappa,\id_G} \otimes \delta_{\kappa,\id_G} \ & \in  \
\C\big[\Out (G)\big]\otimes {L}^2\big(\Out (G)\big)\otimes {L}^2\big(\Out (G)\big) \ , 
\end{align}
where we identify the complex vector space $\C[\Out(G)]$ generated by the
elements of $\Out(G)$ with the dual of 
${L}^2(\Out (G))$. 

To compute the partition function in the presence of a defect network $D$
containing only trivalent vertices with one or two ingoing edges, we proceed as follows.
We pick a triangulation agreeing with the defect network and integrate over all 
lattice gauge fields as
\begin{align}
W_{\textrm{\tiny YM}}\big(\Sigma,G\rtimes\Out(G), e^2\,a;D\big) \coloneqq
  \frac{1}{|\Out(G)|^{N}} \ \sum_{(\kappa_\gamma)\in\Out(G)^{\times|\Sigma^{\textrm{\tiny(1)}}|}} \ \int_{G^{\times|\Sigma^{\textrm{\tiny(1)}}|}} \ & \prod_{\gamma\in
  \Sigma^{\textrm{\tiny(1)}}} \, \dd g_{\gamma} \
  \Wc_D\big((g_\gamma,\kappa_\gamma)\big) \\ & \times \ \prod_{w\in
  \Sigma^{\textrm{\tiny(2)}}} \,
  {\mit\Gamma}\big(\mathcal{U}_w,e^2\,a_w\big) \ ,
\end{align}
where $a=\sum_{w\in\Sigma^{\textrm{\tiny(2)}}}\, a_w$, the sum is
over (flat) $\Out(G)$-bundles on the triangulation of the 
surface $\Sigma$, $N$ is an appropriate normalisation power, and
${\mit\Gamma}(\mathcal{U}_w,e^2\,a_w)$ is the same local function
\eqref{eq:mitGammadef} as for the Yang-Mills theory based on the gauge
group $G$,\footnote{As in Section~\ref{Sec:Lattice}, this is the case
  because locally both gauge theories agree.} but now with the
holonomies $\mathcal{U}_w$ computed for the $G\rtimes \Out(G)$-bundle with parallel transport $(g_\gamma,\kappa_\gamma)$ along
the edges $\gamma\in\partial w$.
The quantity $\Wc_D\big((g_\gamma,\kappa_\gamma)\big) $ is the value of the corresponding Wilson line observable for the $G\rtimes \Out(G)$-bundle described
by the elements $(g_\gamma,\kappa_\gamma)\in G\rtimes\Out(G)$ for
$\gamma\in \Sigma^{\textrm{\tiny(1)}}$, which can be computed as follows:
Combining the junction fields \eqref{Eq. Junction fields} for all vertices defines an element in 
$\C[\Out(G)]^{\otimes|\Sigma^{\textrm{\tiny(1)}}|}\otimes {L}^2(\Out(G))^{\otimes|\Sigma^{\textrm{\tiny(1)}}|}$. 
To produce a complex number from this we act on the elements of
$\C[\Out(G)]$ with the group element of the corresponding edge, and
then apply to it the function in $L^2(\Out(G))$ corresponding to the
endpoint of the edge; this defines $\Wc_D$.   
The form of the junction fields \eqref{Eq. Junction fields} implies that an edge $\gamma_{x,y}$
between two vertices $x,y\in\Sigma^{\textrm{\tiny(0)}}$ induces a delta-function between the sums for the
different vertices of the form $\delta_{\kappa_{\gamma_{x,y}}\, \kappa_x, \kappa_y}$. This implies that $\Wc_D$ is 
non-zero if and only if the parallel transport around every loop has
trivial part in $\Out(G)$. 
In this case we can apply a gauge transformation to set all $\kappa_\gamma=\id_G$. 
Restricting to elements with all $\kappa_\gamma=\id_G$ cancels the factor $|\Out(G) |^N$
in the partition function $W_{\textrm{\tiny YM}}$. Hence we are left with the partition function for 
Yang-Mills theory with gauge group $G$, showing that the reverse orbifold
theory is the Yang-Mills theory we started with:
\begin{align}
W_{\textrm{\tiny YM}}\big(\Sigma,G\rtimes\Out(G), e^2\,a;D\big) =
Z_{\textrm{\tiny YM}}\big(\Sigma,G, e^2\,a\big) \ .
\end{align}

\section{Generalised orbifold of functorial defect Yang-Mills theory}\label{Sec:Funct}

In this section we gauge the $\Out(G)$-symmetry of two-dimensional Yang-Mills theory
using the generalised orbifold construction 
\cite{FFSR:2009orb,Carqueville:2012orb,Brunner:2013orb,Carqueville:2017orb} 
of a functorial defect quantum field theory.
We start by recalling the notion of area-dependent quantum field theories and their state sum constructions in
Section~\ref{Sec:State-Sum-Intro}. Then we give a detailed description of the
bicategory of topological defects of two-dimensional Yang-Mills theories in Section~\ref{Sec:Defect-Bicategory}, and in Section~\ref{Sec:bimodule-from-Hopf} we discuss the regularised Frobenius algebras constructed from a Lie group and its outer automorphism group. 
Finally, in Section~\ref{Sec:Orbifolding-Out-G} we gauge the
$\Out(G)$-symmetry, and using an orbifold equivalence in the orbifold
completion of the topological defect bicategory we give the defect for the reverse orbifold in Section~\ref{Sec:Backwards-Orbifold}.

\subsection{State sum area-dependent quantum field theory with defects}\label{Sec:State-Sum-Intro}

We begin by briefly reviewing the state sum construction of two-dimensional area-dependent quantum field theory with defects \cite{Runkel:2018aqft}. We define area-dependent quantum field theories in the spirit of \cite{Atiyah:1988tft,Segal:1988cft,Segal:1988mf} as symmetric monoidal functors from a bordism category into an appropriate target category. Then we recall some details of the state sum construction, 
and discuss transmissive defects which are the topological defects in area-dependent theories. 

An \emph{area-dependent quantum field theory} is a symmetric monoidal functor from the 
category of two-dimensional bordisms with area $\Bord{\textrm{area}}$ to the category of Hilbert spaces $\Hilb$.
In the former category the objects are disjoint unions of oriented
circles, and the morphisms are oriented bordisms up
to diffeomorphism together with
a positive (and possibly zero for cylinders) real number assigned to each connected component, which we think of as an area. The morphism sets naturally come with a topology induced by the areas of the connected components of surfaces. 
In the category of Hilbert spaces one can choose many different topologies on morphism sets,
but for our purposes the strong operator topology will be relevant. 
For an area-dependent quantum field theory, in addition to being symmetric monoidal, 
we require the bounded linear operators assigned to bordisms to be continuous in the area parameters. 

Area-dependent quantum field theories (without defects) are completely
defined by a variation of the notion of a Frobenius algebra,
analogously to two-dimensional topological field theories. A \emph{regularised Frobenius algebra} consists of a Hilbert space $A$ together with families of maps
$\mu_a:A\otimes A\longrightarrow A$ (products), $\eta_a:\Cb\longrightarrow A$ (units), 
$\Delta_a:A\longrightarrow A\otimes A$ (coproducts) and $\varepsilon_a:A\longrightarrow \Cb$ (counits) which are
continuous in the parameter $a\in\Rb_{>0}$ with respect to the strong operator topology. 
These are required to satisfy parameterised versions of 
associativity, unitality, coassociativity, and counitality:
\begin{align}
	&
	\begin{aligned}
	\mu_{a}\circ( \mu_b\otimes \id_A )&=
	\mu_{a'}\circ( \id_A\otimes \mu_{b'})\ ,
	\quad&
	\mu_{a}\circ\left( \eta_b\otimes \id_A \right)&=
	\mu_{a'}\circ\left( \id_A\otimes \eta_{b'} \right)=:P_{a+b}\ ,\\[4pt]
	\left( \Delta_b\otimes \id_A \right)\circ\Delta_{a}&=
	\left( \id_A\otimes \Delta_{b'} \right)\circ\Delta_{a'}\ ,
	\quad&
	\left( \varepsilon_b\otimes \id_A \right)\circ\Delta_{a}&=
	\left( \id_A\otimes \varepsilon_{b'} \right)\circ\Delta_{a'}=P_{a+b}\ ,
	\end{aligned}
	\label{eq:RFA-rel1}
\end{align}
and of the Frobenius relation
\begin{align}
\Delta_a\circ\mu_b=
	(\id_A\otimes\mu_{b'})\circ(\Delta_{a'}\otimes\id_A)=
	(\mu_{b'}\otimes\id_A)\circ(\id_A\otimes\Delta_{a'}) \ ,
	\label{eq:RFA-rel2}
\end{align}
for all parameters $a,a',b,b'\in\Rb_{>0}$ with $a+b=a'+b'$,
where the map $P_a:A\longrightarrow A$ satisfies
\begin{align}
\lim_{a\to0} \, P_a=\id_A
\end{align}
in the strong operator topology. 

We will heavily rely on the graphical calculus for (strict) symmetric
monoidal categories, in order to simplify the presentation of our calculations. 
We present a morphism $f:A\longrightarrow B$ as 
\begin{align}
\begin{tikzpicture}
	\begin{pgfonlayer}{nodelayer}
		\node [style=square] (0) at (0, 0) {$f$};
		\node [style=none, label={above:$B$}] (1) at (0, 1) {};
		\node [style=none, label={below:$A$}] (2) at (0, -1) {};
	\end{pgfonlayer}
	\begin{pgfonlayer}{edgelayer}
		\draw (1.center) to (0);
		\draw (0) to (2.center);
	\end{pgfonlayer}
\end{tikzpicture}

\end{align}
and the identity morphisms with a straight line 
\begin{align}
\begin{tikzpicture}
	\begin{pgfonlayer}{nodelayer}
		\node [style=none, label={below:$A$}] (2) at (0.5, -1) {};
		\node [style=none] (9) at (-1, 0) {$\id_A=$};
		\node [style=none, label={above:$A$}] (10) at (0.5, 1) {};
	\end{pgfonlayer}
	\begin{pgfonlayer}{edgelayer}
		\draw (10.center) to (2.center);
	\end{pgfonlayer}
\end{tikzpicture}

\end{align}
Composition corresponds to stacking,
the tensor product of objects and morphisms is 
\begin{align}
\input{graphcalc-tensor.tikz} 
\end{align}
and the symmetric braiding is denoted by a crossing 
\begin{align}
\input{graphcalc-sym.tikz}
\end{align}
For more details on this graphical calculus, see for example~\cite{Kassel:1994qg}.

The structure maps of a  regularised Frobenius algebra $A$ are presented as 
\begin{align}
	\input{RFA-str-maps.tikz}
	\label{eq:RFA-str-maps}
\end{align}
and the relations \eqref{eq:RFA-rel1} and \eqref{eq:RFA-rel2} are
\begin{align}
	\input{RFA-relations.tikz}
	\label{eq:RFA-relations}
\end{align}

A regularised Frobenius algebra $A$ is \emph{commutative} if $\mu_a\circ\sigma_{A,A}=\mu_a$
for every $a\in\Rb_{>0}$:
\begin{align}
\input{RFA-comm.tikz}
\end{align}
Area-dependent quantum field theories are classified by commutative  regularised Frobenius algebras:
the underlying Hilbert space is the value of the quantum field theory on the circle $S^1$, and 
the structure maps are the values on the generating morphisms of $\Bord{\textrm{area}}$, which are the cups, caps and pairs of pants. For further details see \cite[Section~3.2]{Runkel:2018aqft}.

An \textsl{area-dependent quantum field theory with defects} is a symmetric monoidal functor from the category
of two-dimensional bordisms with area and defects. In this bordism category we endow manifolds with a stratification, which is a collection of immersed manifolds of lower dimension. The surface components are assigned individual areas and the functor is required to be continuous in all of these area parameters.

The category of bordisms with area and defects comes with three label sets $D_2$, $D_1$ and $D_0$, which respectively label the submanifolds of dimension two, one and zero. The elements of $D_2$ are 
called phases, the elements of $D_1$ are called domain walls or defect conditions, and
the elements of $D_0$ are called junction field labels. For more details see for example \cite{Davydov:2011dt} and \cite[Section~3.3]{Runkel:2018aqft}.

Consider an area-dependent quantum field theory with defects $\funZ$.
A defect line labeled with $x\in D_1$ is \textsl{transmissive} if the value of $\funZ$
on surfaces involving defects labeled with $x$ depends only on the sum of the areas of the surface components separated by the defect; in other words, area can be transmitted through the defect line. These are the topological defects in area-dependent quantum field theories.
When only considering topological defects, the sets $D_0$, $D_1$ and $D_2$ can be organised
into a bicategory using the functor $\funZ$, see Section~\ref{Sec:Defect-Bicategory} below for further
details.

One way to construct examples of area-dependent quantum field theory with defects is using the 
`state sum construction'. Here one works with an appropriate cell decomposition of the surface; for example, faces are allowed to be intersected by defect lines (without junctions) 
at most once, and junctions and faces should intersect at most once.

The set labeling surface components $D_2$ is a set of {strongly separable symmetric 
 Frobenius algebras}.
A regularised Frobenius algebra $A$ is \emph{symmetric} if the natural bilinear pairings
$\varepsilon_a\circ\mu_b:A\otimes A\longrightarrow \Cb$ are symmetric:
$\varepsilon_{a}\circ\mu_{b}\circ\sigma_{A,A}=\varepsilon_{a}\circ\mu_{b}$,
and it is \emph{strongly separable} if there exist algebra homomorphisms $\tau_a:A\longrightarrow A$ which satisfy
$\tau_a\circ\mu_b\circ\Delta_c=\mu_b\circ\Delta_c\circ\tau_a=P_{a+b+c}$
for every $a,b,c\in\Rb_{>0}$. Using such Frobenius algebras ensures
that the state sum construction will be independent of the choice of
cell decomposition.
The examples of  Frobenius algebras considered in this paper are strongly separable symmetric with $\tau_a=P_a$.

Before we can describe the set $D_1$, we need to define bimodules.
A \textsl{bimodule} over  regularised Frobenius algebras $A$ and $B$ is a Hilbert space $X$ together with a family of maps $\rho^X_{a,b}:A\otimes X \otimes B\longrightarrow X$ (the two-sided actions), which we 
denote by 
\begin{align}
\input{bimodule-action.tikz}
\end{align}
satisfying a parameterised version of associativity, and the map
\begin{align}
\input{q-def.tikz}
\end{align}
 satisfies
 \begin{align}
 \lim_{a,b\to0} \, Q_{a,b}^X=\id_X \ .
 \end{align}
One can similarly define left and right modules, and commuting left and right actions
define a bimodule. The converse is not true in general, but the
bimodules considered in this paper are in fact left and right modules,
with corresponding morphisms $Q_a^X$, and hence in the following we only consider such bimodules.

An $A$--$B$-bimodule $X$ is \textsl{dualisable} if there exists a
$B$--$A$-bimodule $\bar{X}$ together with two families of morphisms 
$\beta_{a,b}^X:X\otimes\bar{X}\longrightarrow\Cb$ and 
$\gamma_{a,b}^X:\Cb\longrightarrow\bar{X}\otimes X$, which we denote as
\begin{align}
 \input{bimodule-duality-morph.tikz} 
 \end{align}
 that satisfy the duality relations
\begin{align}
	\input{bimodule-duality-relation.tikz}
	\label{eq:bimodule-duality-relation}
\end{align}
and which are compatible with the action:
\begin{align}
	\input{bimodule-duality-compatibility.tikz}
	\label{eq:bimodule-duality-compatibility}
\end{align}
The set labeling defect lines $D_1$ is a set of dualisable bimodules over the  regularised Frobenius algebras in $D_2$. 
A bimodule is \textsl{transmissive} if the action depends only on the sum of the parameters.
Transmissive bimodules correspond to transmissive defect lines.

In order to give the set $D_0$, we need some more notions.
Let $X$ be an $A$--$B$-bimodule and $Y$ a $B$--$C$-bimodule for regularised Frobenius algebras $A$, $B$ and $C$.
The \textsl{relative tensor product} $X\otimes_B Y$ of $X$ and $Y$
is an $A$--$C$-bimodule which is a coequaliser of the morphisms 
\begin{align}
\input{lr-action.tikz}
\end{align}
If $A$, $B$ and $C$ are strongly separable symmetric Frobenius algebras then the 
relative tensor product is the image of the idempotent 
\begin{align}
\input{rel-tensor-prod-idempot.tikz}
\end{align}
which exists for dualisable bimodules, and the action is given by 
\begin{align}
	\input{rel-tensor-prod-action.tikz}
	\label{eq:rel-tensor-prod-action}
\end{align}
where $\pi$ and $\iota$ are the projection and embedding of the image of the idempotent.
For bimodules which are left and right modules as well, the limit in \eqref{eq:rel-tensor-prod-action} exists as we are allowed to set $b_i=0$.
The fusion of defect lines in the state sum construction corresponds to the 
relative tensor product of bimodules~\cite[Theorem~4.20]{Runkel:2018aqft}.

Similarly, we define the \textsl{cyclic tensor product} $\ctimes_A X$ of an
$A$--$A$-bimodule $X$ 
by identifying the two actions. Instead of giving details here,
we just note that the idempotent with image $\ctimes_A X$ is given by 
\begin{align}
\input{cyclic-tensor-prod-idempot.tikz}
\end{align}
and refer to \cite{Runkel:2018aqft} for further details.

Consider a boundary circle of a surface with defect lines, some of which start or end on
this circle. By \cite[Theorem~4.19]{Runkel:2018aqft}, 
the state space assigned to this circle is
\begin{align}
	\funZ(S^1,A_1,\dots,A_n;X_1,\dots,X_n)=\
	\ctimes_{A_1}X_{1}^{\epsilon_1}\otimes_{A_2}X_{2}^{\epsilon_2}\otimes_{A_3} \dots\otimes_{A_n}X_{n}^{\epsilon_n}\ ,
	\label{eq:state-sum-state-space}
\end{align}
where $A_i\in D_2$, $X_i\in D_1$, and $\epsilon_i\in\left\{ \pm \right\}$ 
depending on the orientation of the $i$-th defect, with
$X_i^{+}=X_i$ and $X_i^{-}=\bar{X}_i$ the dual of $X_i$.

The set $D_0$ of junction field labels is given by families of elements in the state spaces \eqref{eq:state-sum-state-space} which are invariant under the action of cylinders
over the circles, which are cylinders with parallel defect lines. We give more detail on this in Section~\ref{Sec:Defect-Bicategory} below.

Now we sketch what the state sum construction assigns to a surface $\Sigma:S\longrightarrow T$ with defect lines.
We pick a cell decomposition of the surface such that the defect lines intersect only edges and they intersect each edge at most once. We require that every face contains at most one junction of defect lines. 
Then we define $\funZ(\Sigma)$ in two steps.
First we consider the surface $\Sigma'$ obtained from $\Sigma$ by cutting out small disks near 
the junctions, and we regard the new boundary components as ingoing.
Then we compose $\funZ(\Sigma')$ with $\id_{\funZ(S)}$ tensored with the elements from $D_0$ that label the junction fields.

It remains to show how to define $\funZ(\Sigma')$, which is
the value of the functor $\funZ$ on surfaces without junctions of defects. To each face we assign the morphism
\begin{align}
	\input{face-and-its-weight.tikz}
	\label{eq:face-and-its-weight} 
\end{align}
Then we use the duality morphisms of the bimodules, and the morphisms 
$\varepsilon_{a_1}\circ\mu_{a_2}$, to contract legs corresponding to inner edges 
according to the cell decomposition (which describes how faces are
glued together along the edges indicated by dashed lines)
and to define ingoing edges.
Finally we compose with the corresponding projections and embeddings to the state space.

As a detailed computation, consider the cylinder with parallel defect
lines illustrated in Figure~\ref{fig:cylinder-defects}.
\begin{figure}[htb]
\small
	\centering
	\input{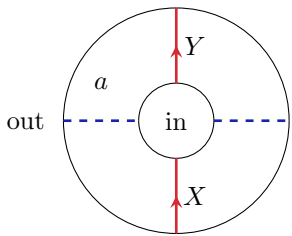}
	\caption{\small A cylinder of total area $a$ with two parallel defect lines
          labeled by $X,Y\in D_1$. 
	The dashed lines indicate a cell decomposition of this surface.}
	\label{fig:cylinder-defects}
\normalsize
\end{figure}
For the two faces we have the two morphisms from \eqref{eq:face-and-its-weight},
for the two dashed edges we contract the legs using the morphisms $\varepsilon_a\circ\mu_{a'}$
and we pull down two legs using the duality morphisms. Finally we
compose with the embedding $\iota$ and projection $\pi$ onto the
cyclic and relative tensor products of $X$ and $Y$ to get
\begin{align}
	\input{cylinder-defects-computation.tikz}
	\label{eq:cylinder-defects-computation}
\end{align}
In the first step we used the properties of the duality morphisms, and in the second step
the definitions of $\iota$ and $\pi$.
Here and in the following we do not write the area parameters
explicitly in order to streamline the presentation, and since we can distribute the area parameters among the morphisms arbitrarily.
We will also not write the morphisms $Q_{a,b}^X$ explicitly.
In Section~\ref{Sec:Orbifolding-Out-G} we give computations which involve defect junctions as well.

\subsection{The defect bicategory of Yang-Mills theory}
\label{Sec:Defect-Bicategory}

For the remainder of this paper we focus on the area-dependent quantum
field theory $\Zc=\mathcal{Z}_{\textrm{\tiny YM}}$ corresponding to
two-dimensional Yang-Mills theory, as defined in
Section~\ref{sec:intro}; in this case the state sum construction
provides a rigorous implementation of the lattice regularisation of 
Section~\ref{Sec: Physics}. Compared to Section~\ref{Sec: Physics}, in
the following we set the gauge coupling constant to $e=1$ without loss
of generality. The weak-coupling limit, which determines a topological
field theory, is then equivalent to the zero
area limit $a\longrightarrow0$.

We define a bicategory of topological defects $\Bscr_{\textsf{\tiny YM}}$ in the spirit of \cite{Davydov:2011dt,Carqueville:2012orb}.
This bicategory has as objects regularised Frobenius algebras of the form $A=L^2(G)$ where $G$ is a 
compact semi-simple Lie group. The 1-morphisms $X:A\longrightarrow B$ are transmissive bimodules with duals 
and the composition is given by the relative tensor product. 
The 2-morphisms $X\Longrightarrow Y$ for $X,Y:A\longrightarrow B$ are given by the set of families 
of maps $\{\phi_{a}:\Cb\longrightarrow\ \ctimes_A Y\otimes_B \bar{X}\}_{a\in\Rb_{>0}}$ 
which are invariant under the action of cylinders:\footnote{This
  definition is an incarnation of the operator-state correspondence of
local quantum field theory.}
\begin{align}
	\Cc_b \circ\phi_{a}=\phi_{a+b} \ ,
	\label{eq:invariant-family}
\end{align}
where $\Cc_a$ is the value of $\Zc_{\textrm{\tiny YM}}$ on the
cylinder illustrated in Figure~\ref{fig:cylinder-defects} and it is computed in \eqref{eq:cylinder-defects-computation}.
We write $\sfH^{\mathrm{inv}}(Y \otimes_B \bar{X})$ for this set of invariant families.
For any morphism $\phi:\Cb\longrightarrow\ \ctimes_A Y\otimes_B
\bar{X}$ we can define a family $\phi_a:=\Cc_a \circ\phi$, and
in this case $\phi_a\longrightarrow\phi$ in the limit $a\longrightarrow0$. 
However, there exist invariant families for which this limit does not
exist; an example of such a family is
\begin{align}
\big\{\eta^{C\ell^2(G)}_a:\Cb\longrightarrow\ \ctimes_{L^2(G)}L^2(G)\otimes_{L^2(G)} L^2(G)
\simeq \Cl(G)\big\}_{a\in\Rb_{>0}} \ .
\end{align}

Vertical composition of 2-morphisms is given by the pair of pants 
\begin{align} \label{eq:pair-of-pantsc}
	\input{pants-vertical.tikz}
\end{align}
of total area $c$. Explicitly, the vertical composition of $\left\{ \phi_a \right\}_{a\in\Rb_{>0}}:X\Longrightarrow Y$
and $\left\{ \varpi_b \right\}_{b\in\Rb_{>0}}:Y\Longrightarrow Z$ for
$X,Y,Z:A\longrightarrow B$ is the family
\begin{align}
	\left\{ \phi_a \right\}_{a\in\Rb_{>0}}
	\circ^\mathrm{ver}
	\left\{ \varpi_b \right\}_{b\in\Rb_{>0}}:=
	\big\{ \mu_c^\mathrm{ver}\circ(\phi_a\otimes \varpi_b)
  \big\}_{a+b+c\in\Rb_{>0}}:X\Longrightarrow Z \ ,
	\label{eq:vertical-composition}
\end{align}
where $\mu_c^\mathrm{ver}$ is the value of $\Zc_{\textrm{\tiny YM}}$
on the pair of pants \eqref{eq:pair-of-pantsc}.
The unit of this product is given by the value of $\Zc_{\textrm{\tiny YM}}$ on a disk crossed by a defect line.

Horizontal composition of 2-morphisms is given by acting with the pair of pants
\begin{align} \label{eq:pair-of-pantsh}
	\input{pants-horizontal.tikz}
\end{align}
of total area $c$. Explicitly, the horizontal composition of $\left\{ \phi_a \right\}_{a\in\Rb_{>0}}:X\Longrightarrow Y$ and 
$\left\{ \phi'_b \right\}_{b\in\Rb_{>0}}:X'\Longrightarrow Y'$ for
$X,Y:A\longrightarrow B$ and $X',Y':B\longrightarrow C$ is the family
\begin{align}
	\left\{ \phi_a \right\}_{a\in\Rb_{>0}}
	\circ^\mathrm{hor}
	\left\{ \phi'_b \right\}_{b\in\Rb_{>0}}:=
	\big\{ \mu_c^\mathrm{hor}\circ(\phi_a\otimes\phi'_b) \big\}_{a+b+c\in\Rb_{>0}}
	:X\otimes_{B}X'\Longrightarrow Y\otimes_{B}Y' \ ,
	\label{eq:horizontal-composition}
\end{align}
where $\mu_c^\mathrm{hor}$ is the value of $\Zc_{\textrm{\tiny YM}}$
on the pair of pants \eqref{eq:pair-of-pantsh}.
The unit of this composition is the value of $\Zc_{\textrm{\tiny YM}}$
on a disk with trivial defect line.

\begin{lemma}
	The morphisms $\mu_c^\mathrm{ver}$ and $\mu_c^\mathrm{hor}$
	for the vertical and horizontal compositions are given by
	\begin{align}
		\input{mu-vertical-horizontal.tikz}
		\label{eq:mu-vertical-horizontal}
	\end{align}
	The unit of $\mu_c^\mathrm{ver}$ is
        $\big\{\pi\circ\mathrm{coev}_a^X:\Cb\longrightarrow\ \ctimes_B
        \bar{X}\otimes_A X\big\}_{a\in\Rb_{>0}}$, while the unit of
        $\mu_c^\mathrm{hor}$ is $\big\{\iota\circ\eta^A_a:\Cb\longrightarrow\ \ctimes_A A\otimes_A A
        \big\}_{a\in\Rb_{>0}}$.
	\label{lem:mu-vertical-horizontal}
\end{lemma}
\begin{proof}
	Consider the cell decompositions
	\begin{center}
		\input{pants-vertical-decomp.tikz}
		\qquad
		and 
		\qquad
		\input{pants-horizontal-decomp.tikz}
	\end{center}
	To the first one the state sum construction assigns the morphism
	\begin{align}
		\input{horizontal-composition-calc.tikz}
		\label{eq:horizontal-composition-calc}
	\end{align}
	which is $\mu_c^\mathrm{ver}$ after simplifying the expression
        using the definitions of $\iota$ and $\pi$ as in the
        calculation of \eqref{eq:cylinder-defects-computation}.
 For the second pair of pants the morphism is
	\begin{align}
		\input{vertical-composition-calc.tikz}
		\label{eq:vertical-composition-calc}
	\end{align}
	which can be similarly disentangled to give $\mu_c^\mathrm{hor}$.
\end{proof}

Let $\Hom_{A|B}^\mathrm{fam}(X,Y)$ denote the families of bimodule morphisms 
$\{\psi_a:X\longrightarrow Y\}_{a\in\Rb_{>0}}$ which satisfy the invariance property
\begin{align}
	Q_{b}^Y\circ\psi_a=\psi_a\circ Q_{b}^X=\psi_{a+b} \ .
	\label{eq:module-morphism-invariance}
\end{align}
The composition of two families of bimodule morphisms
$\{\xi_a:Y\longrightarrow Z\}_{a\in\Rb_{>0}}\in \Hom_{A|B}^\mathrm{fam}(Y,Z)$ and
$\{\psi_b:X\longrightarrow Y\}_{b\in\Rb_{>0}}\in \Hom_{A|B}^\mathrm{fam}(X,Y)$ is defined via the pointwise composition
\begin{align}
	\left\{ \xi_a \right\}_{a\in\Rb_{>0}}\circ\left\{ \psi_b \right\}_{b\in\Rb_{>0}}:=
	\left\{ \xi_a\circ\psi_b \right\}_{a+b\in\Rb_{>0}}  \ \in \
  \Hom_{A|B}^\mathrm{fam}(X,Z) \ .
	\label{eq:composition-of-families-of-morphisms}
\end{align}
This composition is well defined, as it is independent of the choice of parameters $a$ and $b$.
The unit for this composition is the family defined via the 
identity morphisms of the respective bimodules.
Similarly we define the relative tensor product 
of two families $\left\{ \psi_a:X\longrightarrow Y
\right\}_{a\in\Rb_{>0}}\in \Hom_{A|B}^\mathrm{fam}(X,Y)$ and 
$\left\{ \psi'_b:X'\longrightarrow Y' \right\}_{b\in\Rb_{>0}}\in \Hom_{B|C}^\mathrm{fam}(X',Y')$
pointwise by
\begin{align}
	\left\{ \psi_a \right\}_{a\in\Rb_{>0}}\otimes_B 
	\left\{ \psi'_b \right\}_{b\in\Rb_{>0}}:=
	\left\{ \psi_a\otimes_B\psi'_b \right\}_{a+b\in\Rb_{>0}} \ \in
  \ \Hom_{A|C}^\mathrm{fam}(X\otimes_B X', Y\otimes_B Y') \ ,
	\label{eq:tensor-product-of-families-of-morphisms}
\end{align}
which is again well defined.
The units for this tensor product are the families $\left\{
  P_a^A:A\longrightarrow A \right\}_{a\in\Rb_{>0}}$ in $\Hom_{A|A}^{\mathrm{fam}}(A,A)$.

We then have an analogue of \cite[Lemma~3.9]{Davydov:2011dt} given by
\begin{lemma}\label{lem:FSmaps}
	The two maps
	\begin{equation}
		\begin{tikzcd}
			\sfH^{\mathrm{inv}}(\bar{X}\otimes_A Y)\ar[bend left = 25]{r}{\Fc}& 
			\Hom^\mathrm{fam}_{A|B}(X,Y)\ar[bend left =
                        25]{l}{\Sc} \ ,
		\end{tikzcd}
		\label{eq:states-vs-bimodule-morphisms-iso}
	\end{equation}
	given by
	\begin{align}
		\input{F-S-maps.tikz}
		\label{eq:F-S-maps}
	\end{align}
	for $\{\phi_a\}_{a\in\Rb_{>0}}\in \sfH^{\mathrm{inv}}(\bar{X}\otimes_A Y)=\Hom_{\Cb}(\Cb,\sfH^{\mathrm{inv}}(\bar{X}\otimes_A Y))$ and $\{\psi_a\}_{a\in\Rb_{>0}}\in \Hom^\mathrm{fam}_{A|B}(X,Y)$,
	are inverse to each other.
	$\Fc$ sends vertical compositions to compositions of families of bimodule morphisms,
	 horizontal compositions to the relative tensor product of families of 
	bimodule morphisms, and units to units.
	\label{lem:2-morphisms-in-B}
\end{lemma}
\begin{proof}
	We first look at the composition $\Sc\circ \Fc$:
	\begin{align}
		\input{SF-comp.tikz}
		\label{eq:SF-comp}
	\end{align}
	where we used the definition of $\Cc_{b+c}$ and the invariance property of 
	the family $\left\{ \phi_a \right\}_{a\in\Rb_{a>0}}$.
	Then we look at the composition $\Fc\circ \Sc$:
	\begin{align}
		\input{FS-comp.tikz}
		\label{eq:FS-comp}
	\end{align}
	where we used the definition of the projector $\Dc_0$, the compatibility
	of the duality morphisms with the actions, and the invariance property.

	Next we show the compatibility of $\Fc$ and $\Sc$ with the vertical (horizontal)
	composition of invariant families and the composition (relative tensor product)
	of families of morphisms.
	For the remainder of this proof we do not write out the parameters of the families.
	For the vertical composition we have
	\begin{align}
		\input{S-vertical-compat.tikz}\ .
		\label{eq:S-vertical-compat}
	\end{align}
	For the horizontal composition we have
	\begingroup
\allowdisplaybreaks
		\begin{align}
		&\input{S-horizontal-compat-1.tikz}\nonumber\\[4pt]
		& \hspace{2.5cm} \input{S-horizontal-compat-2.tikz}
		\label{eq:S-horizontal-compat}
	\end{align}
	\endgroup
	where we used the cyclicity of the tensor product
	$\ctimes_B\bar{X}\otimes_A Y\otimes_B Y'\otimes_C \bar{X'}$
	together with the fact that the relative tensor product of morphisms can be expressed
	using the projections $\pi$ and embeddings~$\iota$.
\end{proof}

Lemma~\ref{lem:FSmaps} gives two equivalent ways of thinking about the 2-morphisms in the 
bicategory of topological defects.
Sometimes it is easier to work with families of bimodule morphisms as their
weak-coupling limit exists more frequently, for example the families $\{\eta_a^A\}_{a\in\Rb_{>0}}$,
$\big\{ Q_{a}^X \big\}_{a\in\Rb_{>0}}$ and
$\big\{ P_a^{A} \big\}_{a\in\Rb_{>0}}$, and we can compute with the
limits instead of the families.
On the other hand, in general the weak-coupling limit of an invariant family
may not exist: Take for example the Frobenius algebra $A=\bigoplus_{k\in\N}\,\Cb\,e_k$ with orthonormal basis $\left\{ e_k \right\}_{k\in\N}$,
product $\mu(e_j\otimes e_k)=\delta_{jk} \, e_k$ and unit
$\sum_{k\in\N}\, \e^{-a\,k^2}\,e_k$, and consider $A$ as a
bimodule over itself. Then the family of endomorphisms
$\left\{ \phi_a \right\}_{a\in\Rb_{>0}}$ of $A$ given by $\phi_a(e_k)=\e^{-a\,k^2}\,k^2\,e_k$
clearly does not have a weak-coupling limit.

Accordingly we can now give a working definition.
\begin{definition}
	The \textsl{topological defect bicategory $\Bscr_{\textsf{\tiny YM}}$ of
          two-dimensional Yang-Mills theories} has:
	\begin{itemize}
		\item[(a)] \ \underline{Objects:} Hilbert spaces $L^2(G)$ for $G$ compact
                  semi-simple Lie groups with regularised Frobenius algebra structure given by \eqref{eq:L2G-RFA-3};
		\item[(b)] \ \underline{1-morphisms:} Transmissive bimodules with duals
                  between regularised Frobenius algebras $L^2(G)$ and
                  $L^2(H)$; and
		\item[(c)] \ \underline{2-morphisms:} Invariant families of bimodule morphisms.
	\end{itemize}
	\label{def:defect-bicat}
\end{definition}

In order to apply techniques from \cite{Carqueville:2012orb} later on
we will need
\begin{proposition}
	The bicategory $\Bscr_{\textsf{\tiny YM}}$ is idempotent complete; that is, its morphism categories are idempotent complete.
	\label{prop:B-idempotent-complete}
\end{proposition}
\begin{proof}
	Let $\psi=\left\{ \psi_a: X\longrightarrow X \right\}_{a\in\Rb_{>0}}$ be an idempotent on
        an $A$--$B$-bimodule $X$, that is it obeys
	$\psi_a\circ\psi_b=\psi_{a+b}$.
	Let $Y$ be the closure of the subspace $\bigcup_{a\in\Rb_{>0}}\,\im(\psi_a)$,
        $p:X\longrightarrow Y$ the projection and 
	$e:Y\longrightarrow X$ the embedding of the subspace $Y\subset X$.
	The Hilbert space $Y$ becomes an $A$--$B$-bimodule via the induced action
	$p\circ\rho^X\circ(\id_A\otimes e \otimes \id_B)$.
	Since $\psi_a$ is an intertwiner, the action on $X$ indeed restricts to $Y$.
	Set $\pi:=\{p_a=p\circ\psi_a:X\longrightarrow Y\}_{a\in\Rb_{>0}}$ and
        $\iota:=\{e_a=\psi_a\circ e:Y\longrightarrow X\}_{a\in\Rb_{>0}}$.
	Then $\iota\circ\pi=\psi$ and $\pi\circ\iota=\id_Y$; the first
        equation is clear from the definition of $p$ and $e$, while
        the second equation follows from $\psi_a(y)=Q_{a}^Y(y)$ for $y\in Y$ by the definition of $Y$.
\end{proof}

\begin{example}\label{ex:bimodule-Wilson}
Wilson lines can be described by bimodules over $L^2(G)$ as follows.
Let $V$ be a representation of $G$ and consider
$V\otimes L^2(G)$ with the commuting left and right actions
\begin{align}
	\begin{aligned}
		\psi\cdot(v\otimes f) =\left[ x\longmapsto\int_G\,\dd
                  y \ \psi(y)\,(y\cdot v)\, f(y^{-1}\,x) \right]
                \qquad \mbox{and} \qquad
	(v\otimes f) \cdot\psi=v\otimes (f*\psi)\ ,
	\end{aligned}
	\label{eq:Wilson-line-action-1}
\end{align}
for $v\in V$ and $\psi,f\in L^2(G)$.
Similarly, $L^2(G)\otimes V$ is a bimodule via
\begin{align}
	\begin{aligned}
		\psi\cdot(f\otimes v)=(\psi*f)\otimes v \qquad
                \mbox{and} \qquad
		(f\otimes v)\cdot\psi=\left[ x\longmapsto\int_G\,\dd y \
                  (y^{-1}\cdot v)\, f(x\,y^{-1})\, \psi(y) \right]\ .
	\end{aligned}
	\label{eq:Wilson-line-action-2}
\end{align}
The dual of $V\otimes L^2(G)$ is $L^2(G)\otimes V^{*}$, where $V^*$ is the
dual of $V$. 
If $G$ is connected then Wilson lines are {\it not} transmissive, 
so they are not 1-morphisms in $\Bscr_{\textsf{\tiny YM}}$.
Nevertheless we will need these for disconnected gauge groups.
For further details on Wilson lines see \cite[Proposition~5.10]{Runkel:2018aqft}.
\end{example}

\begin{example}\label{ex:bimodule-twisted}
The twisted bimodules $L_{\varphi}=L^2(G)$ for $\varphi\in\Out(G)$ with
action given in \eqref{eq:twisted-bimodule-action}
are transmissive bimodules over $L^2(G)$, so they are 1-morphisms in $\Bscr_{\textsf{\tiny YM}}$.
\end{example}

\subsection{Frobenius algebras from symmetry defects}\label{Sec:bimodule-from-Hopf}

Next we define regularised Frobenius algebras and their bimodules in
the category
$\Hilb$, starting from a Lie group and its outer automorphism group. 
	Let $G$ be a compact semi-simple Lie group and
        $\Gamma<\Out(G)$ a subgroup of outer automorphisms of $G$.
	We will sometimes use the notation $L=L^2(G)$, $H=L^2(\Gamma)$
        and $K=L^2(G\rtimes\Gamma)$ for brevity.

	The group $\Out(G)$ is finite and the algebra 
	$L^2(\Gamma)$ is isomorphic to 
	the group algebra of $\Gamma$, which has the structure of a Hopf algebra 
	with coproduct 
	\begin{align}
\input{Hopf-coproduct.tikz} \qquad \phi\longmapsto \Delta_H(\phi)=\big[
          (\gamma,\kappa)\mapsto |\Gamma|\, \phi_\gamma\,\delta_{\gamma,\kappa} \big]
		=:\phi_{(1)}\otimes\phi_{(2)}
		\label{eq:coproduct-Hopf}
	\end{align}
	and antipode
\begin{align}
\begin{tikzpicture}
	\begin{pgfonlayer}{nodelayer}
		\node [style=none] (0) at (0.5, -0.75) {};
		\node [style=white dot big] (1) at (0.5, 0) {};
		\node [style=none] (2) at (0.5, 0.75) {};
		\node [style=none] (3) at (-0.5, 0) {$S=$};
	\end{pgfonlayer}
	\begin{pgfonlayer}{edgelayer}
		\draw (2.center) to (1);
		\draw (1) to (0.center);
	\end{pgfonlayer}
\end{tikzpicture}
 \qquad S(\delta_\gamma)=\frac1{|\Gamma|}\,
  \delta_{\gamma^{-1}} \ ,
\end{align}
where $\delta_\gamma(\kappa)=\delta_{\gamma,\kappa}$. 
	The algebra $L^2(\Gamma)$ acts on $L^2(G)$ via
	\begin{align}
\input{left-action.tikz}	 \qquad \phi\cdot f=\frac{1}{|\Gamma|}\,
          \sum_{\gamma\in\Gamma}\, \phi_\gamma\, f\circ \gamma^{-1}
		\label{eq:Gamma-acts-on-L}
	\end{align}
	for $\phi\in L^2(\Gamma)$ and $f\in L^2(G)$.
	 
Using the
	Hopf algebra structure of $L^2(\Gamma)$
	we endow $L^2(G)\otimes L^2(\Gamma)$ with the structure of a
        regularised Frobenius algebra with unit
	\begin{equation}
		\input{unit-LH.tikz} \ ,
		\label{unit-L2G-o-L2Gamma}
	\end{equation}
	and product
	\begin{align}
\input{semidirect-product.tikz}	\qquad (f\otimes\phi)*(g\otimes\psi)=
		\big(f*(\phi_{(1)}\cdot g)\big)\otimes\big(\phi_{(2)}*\psi\big)
		\label{eq:product-L2G-o-L2Gamma-draw}
	\end{align}
	for $f,g\in L^2(G)$ and $\phi,\psi\in L^2(\Gamma)$.
	We define the coproduct and counit to be the adjoint operators of the product and unit respectively.
	We denote this Frobenius algebra by $L^2(G)\rtimes
        L^2(\Gamma)$.

	The map 
	\begin{align}
\input{L-to-LH.tikz} \qquad L^2(G)\longrightarrow L^2(G)\rtimes L^2(\Gamma) \
          , \quad f\longmapsto f\otimes |\Gamma|\,\delta_{\id_G}
	\label{eq:L2Ginclusion}
	\end{align}
	is a homomorphism of regularised Frobenius algebras. 
	Using this morphism we obtain an $L^2(G)$--$L^2(G)$-bimodule structure on
	$L^2(G)\rtimes L^2(\Gamma)$.

\begin{proposition}\label{prop:bimodule-from-Hopf}
	\begin{enumerate}
		\item \label{prop:bimodule-from-Hopf:1}
	The map
	\begin{align}
		\Phi:L\rtimes H \longrightarrow K \ , \quad 
			f\otimes\phi\longmapsto\big[ (x,\gamma)\mapsto
          f(x)\,\phi_\gamma \big]
	\label{eq:L2G-o-L2Gamma-to-L2G-rtimes-Gamma}
	\end{align}
	is an isomorphism of regularised Frobenius algebras in $\Hilb$. The map \eqref{eq:L2Ginclusion} endows
	$K$ with the structure of a transmissive $L$--$L$-bimodule. 
		\item \label{prop:bimodule-from-Hopf:3}
	The $L$--$L$-bimodule $K$ 
	is a strongly separable symmetric Frobenius algebra 
	in $\Bscr_{\textsf{\tiny YM}}(L,L)$ via the structure morphisms
	\begin{align}
		\bar{\mu}{}^K&:=\big\{ K\otimes_{L} K\xrightarrow{ \ \iota
                           \ } K\otimes K\xrightarrow{ \ \mu_a \ }K
                           \big\}_{a\in\Rb_{>0}} \ , \nonumber \\[4pt]
		\bar{\Delta}{}^K&:=\big\{ K\xrightarrow{ \ \Delta_a \
                              }K\otimes K\xrightarrow{ \ \pi \ }
                              K\otimes_{L}K \big\}_{a\in\Rb_{>0}} \ ,
                              \nonumber \\[4pt]
		\bar{\eta}{}^K&:=\big\{ L\xrightarrow{ \ \bar{\eta}_a \ }
                            K \big\}_{a\in\Rb_{>0}} \ , \quad
                            \bar{\eta}_a(f)=
                            P_a^L(f)\otimes|\Gamma|\,\delta_{\id_{G}}
                            \ , \nonumber \\[4pt]
		\bar{\varepsilon}{}^K&:=\big\{ K\xrightarrow{ \
                                   \bar{\varepsilon}_a \ } L
                                   \big\}_{a\in\Rb_{>0}} \ , \quad
                                   \bar{\varepsilon}_a(f\otimes\phi)=
                                   P_a^L(f) \,
                                   \frac1{|\Gamma|}\,\phi_{\id_{G}}
                                   \ .
		\label{eq:prop:bimodule-from-Hopf}
	\end{align}
\item \label{prop:bimodule-from-Hopf:5}
The $L$--$L$-bimodule $M:=\bigoplus_{\varphi\in\Gamma}\,L_{\varphi}$ is a strongly separable symmetric Frobenius algebra 
	in $\Bscr_{\textsf{\tiny YM}}(L,L)$ via the structure morphisms
	\begin{align}
		\bar{\mu}{}^M&:=\Big\{ M\otimes_{L} M\xrightarrow{ \ \iota
                           \ } M\otimes
                           M\xrightarrow{\sum\limits_{\varphi,\omega\in\Gamma}\,\mu_a^{\varphi,\omega}}M
                           \Big\}_{a\in \Rb_{>0}} \ , \nonumber \\[4pt]
		\bar{\Delta}{}^M&:=\Big\{
                              M\xrightarrow{\frac1{|\Gamma|}\,\sum\limits_{\varphi,\omega\in\Gamma}\,\Delta_a^{\varphi,\omega}}M\otimes
                              M\xrightarrow{ \ \pi \ } M\otimes_{L}M
                              \Big\}_{a\in \Rb_{>0}} \ , \nonumber \\[4pt]
		\bar{\eta}{}^M&:=\Big\{ L\xrightarrow{ \ P^L_a \
                                }L=L_{\id_G}\xhookrightarrow{ \ ~~ \
                                }\bigoplus_{\varphi\in\Gamma}\,L_{\varphi}=M
                                \Big\}_{a\in \Rb_{>0}} \ ,
                                \nonumber \\[4pt]
		\bar{\varepsilon}{}^M&:=\Big\{
                                       M=\bigoplus_{\varphi\in\Gamma}\,L_{\varphi}\twoheadrightarrow
                                       L_{\id_G}=L \xrightarrow{
                                       \ P^L_a \ }L\Big\}_{a\in
                                       \Rb_{>0}} \ ,
		\label{eq:prop:bimodule-from-Hopf-M}
	\end{align}
	where
\begin{align}
\input{mu-ab.tikz} \qquad \input{Delta-ab.tikz} \qquad \begin{tikzpicture}
	\begin{pgfonlayer}{nodelayer}
		\node [style=none, label={above:$L$}] (0) at (0, 1) {};
		\node [style=none, label={below:$L$}] (1) at (0, -1) {};
		\node [style=square] (2) at (0, 0) {$\varphi$};
		\node [style=none] (4) at (2, 0) {$\ =(\varphi^{-1})^{*}$};
	\end{pgfonlayer}
	\begin{pgfonlayer}{edgelayer}
		\draw (0.center) to (1.center);
	\end{pgfonlayer}
\end{tikzpicture}
 \ .
\end{align}
	\item
	The map
	\begin{align}
		\Psi:\bigoplus_{\varphi\in\Gamma}\,L_{\varphi}\longrightarrow
                K \ , \quad 
\sum_{\varphi\in\Gamma}\,f_{\varphi}\longmapsto\big[ (x,\gamma)\mapsto f_{\gamma}\big(\gamma^{-1}(x)\big) \big]
		\label{eq:prop:outG-defect}
	\end{align}
	is an isomorphism of $L$--$L$-bimodules as well as of 
Frobenius algebras in $\Bscr_{\textsf{\tiny YM}}(L,L)$. 
		\label{prop:bimodule-from-Hopf:6}
	\end{enumerate}
\end{proposition}
\begin{proof}
	For Part~\ref{prop:bimodule-from-Hopf:1}, we note that
	 $\Delta_H(\delta_{\varphi})=|\Gamma|\,\delta_{\varphi}\otimes\delta_{\varphi}$
	and $|\Gamma|\,\delta_{\varphi}\cdot f=f\circ\varphi^{-1}$.
	Using
	\eqref{eq:product-L2G-o-L2Gamma-draw}
	we first compute
	\begin{align}
		\Phi\big((f\otimes\delta_{\varphi})*(g\otimes\phi)\big)(x,\omega)&= 
		\Phi\big((f*|\Gamma|\,\delta_\varphi\cdot
                                                                                   g)\otimes(\delta_\varphi*\phi)\big)(x,\omega)
                                                                                   \nonumber
                                                                                   \\[4pt]
          &= 
		\big(f*(g\circ\varphi^{-1})\big)(x) \,
            \frac1{|\Gamma|}\, \phi_{\varphi^{-1}\,\omega} \ .
		\label{eq:LH-mult1}
	\end{align}
	Then we compute
	\begin{align}
		\big(\Phi(f\otimes\delta_{\varphi})*\Phi(g\otimes\phi)\big)(x,\omega)&=
		\int_{G}\, \dd y \ \frac1{|\Gamma|}\, \sum_{\gamma\in\Gamma}\,
		f(y)\,\delta_{\varphi}(\gamma)\,g\big(\gamma^{-1}(y^{-1}\,x)\big)\,\phi_{\gamma^{-1}\,\omega}
                                                                                       \nonumber
                                                                                       \\[4pt]
          &=
		\frac1{|\Gamma|}\,
            \big(f*(g\circ\varphi^{-1})\big)(x)\,\phi_{\varphi^{-1}\,\omega}
            \ ,
		\label{eq:LH-mult2}
	\end{align} 
	which agrees with \eqref{eq:LH-mult1}.

	Next we show that $\Phi\big(\eta_a^{L^2(G)}\otimes|\Gamma|\,\delta_{\id_G}\big)$ is
	a unit for this multiplication, which by uniqueness of the unit of a regularised algebra is precisely $\eta_a^{L^2(G\rtimes\Gamma)}$.
For $F=f\otimes\phi\in K$ we compute
\begin{align}
		\Big(\Phi\big(\eta_a^{L^2(G)}\otimes|\Gamma|\,\delta_{\id_G}\big)*(f\otimes\phi)\Big)(x,\omega)&=
		\int_{G}\,\dd y \ \frac1{|\Gamma|}\, \sum_{\gamma\in\Gamma}\,
		\eta_a^{L^2(G)}(y)\,|\Gamma|\,\delta_{\id_G}(\gamma)\,
		f\big(\gamma^{-1}(y^{-1}\,x)\big)\,\phi_{\gamma^{-1}\,\omega}
                                                                                                                 \nonumber
                                                                                                                 \\[4pt]
&=
		\big(\eta_a^{L^2(G)}*f\big)(x)\,\phi_\omega \ ,
\end{align}
and
\begin{align}
		\Big((f\otimes\phi)*\Phi\big(\eta_a^{L^2(G)}\otimes|\Gamma|\,\delta_{\id_G}\big)\Big)(x,\omega)&=
		\int_{G}\, \dd y \ \frac1{|\Gamma|}\,\sum_{\gamma\in\Gamma}\,
		f(y)\,\phi_\gamma\,\eta_a^{L^2(G)}\big(\gamma^{-1}(y^{-1}\,x)\big)\,|\Gamma|\,\delta_{\id_G}\big(\gamma^{-1}\,\omega\big)
                                                                                                                \nonumber
                                                                                                                \\[4pt]
&=
		\big(f\ast\eta_a^{L^2(G)}\big)(x)\,\phi_\omega \ ,
		\label{eq:LH-unit}
	\end{align}
both of which are equal to $F(x,\omega)$ in the $a\longrightarrow0$ limit, showing that
	$\Phi\big(\eta_a^{L^2(G)}\otimes|\Gamma|\,\delta_{\id_G}\big)$ is the unit of $K$.

	Part~\ref{prop:bimodule-from-Hopf:3} follows from the fact
        that $K$ is a strongly separable symmetric Frobenius algebra and from Part~\ref{prop:bimodule-from-Hopf:1}.
	We only present the computation which shows associativity.
	Let us compute the $a+b+c$ component of the family $\bar{\mu}{}^K\circ(\bar{\mu}{}^K\otimes_L \id_K)$:
	\begin{align}
		\input{mu-bar-associative.tikz}
		\label{eq:mu-bar-associative}
	\end{align}
	where we used associativity of the relative tensor product and 
	 $\mu_c\circ \Dc_0=\mu_c$.
	Computing the $a+b+c$ component of the family 
	$\bar{\mu}{}^K\circ( \id_K \otimes_L \bar{\mu}{}^K)$ in a similar way gives the same family.

For Parts~\ref{prop:bimodule-from-Hopf:5}~and~\ref{prop:bimodule-from-Hopf:6},
        the verification that the isomorphism is an $L$--$L$-bimodule morphism
	follows from a direct computation:
	For the left $L$-action we have
	\begin{align}
		\Psi\Big( g\cdot\sum_{\varphi\in\Gamma}\,f_{\varphi} \Big)&=
		\Psi\Big(
                                                                             \sum_{\varphi\in\Gamma}\,f_{\varphi}*(g\circ\varphi)\Big)
                                                                             \nonumber
                                                                             \\[4pt]
&=
		\Psi\Big( \sum_{\varphi\in\Gamma}\,\big(
  f_{\varphi}\circ\varphi^{-1}\big)*g\Big) \nonumber
          \\[4pt] &=
		\big[ (x,\gamma)\longmapsto \big(
                    (f_{\gamma}\circ\gamma^{-1})*g\big)(x) \big]
                    \nonumber \\[4pt]
&= g\cdot\Psi\Big(\sum_{\varphi\in\Gamma}\,f_{\varphi} \Big) \ .
		\label{eq:plus-L-K-iso1}
	\end{align}
	That $\Psi$ commutes with the right $L$-action can be similarly shown using 
	\begin{align}
		(f\otimes\phi)\cdot g=\big[ (x,\gamma)\longmapsto
          \big(f*(g\circ\gamma^{-1})\big)(x)\, \phi_\gamma \big] \ .
		\label{eq:plus-L-K-iso2}
	\end{align}

 Finally, we compare the Frobenius algebra structure on $M$ defined in
 \eqref{eq:prop:bimodule-from-Hopf-M} with the structure transported from
        $K$ via $\Psi$. Here we show that the product and coproduct agree.
        Since the $a\longrightarrow0$ limits of these families exist, 
        it is enough to check $\Psi\circ P_a^{M}=P_a^{K}\circ\Psi$ 
        and compare the limits; this indeed holds, as it is given by
        scaling basis elements by factors $\e^{-a\,C_2(\alpha)/2}$. For the product we compute
	\begin{align}
		\Psi^{-1}\bigg( \Psi\Big(
          \sum_{\varphi\in\Gamma}\,f_{\varphi} \Big)&*\Psi\Big(
                                                     \sum_{\omega\in\Gamma}\,g_{\omega}
                                                     \big) \bigg)
                                                     \nonumber \\[4pt]
&=
		\Psi^{-1}\bigg( \Big[
  (x,\gamma)\longmapsto\int_{G}\, \dd y \ \frac1{|\Gamma|}\,
  \sum_{\kappa\in\Gamma}\,
  f_{\kappa}\big(\kappa^{-1}(y)\big)\,g_{\kappa^{-1}\,\gamma}\big(\gamma^{-1}\,\kappa\,\kappa^{-1}(y^{-1}\,x)\big)
  \Big] \bigg) \nonumber \\[4pt]
& \hspace{3cm} =
		\Big[ (x,\gamma)\longmapsto\int_{G}\, \dd z \
  \frac1{|\Gamma|}\, \sum_{\sigma\in\Gamma}\,
  f_{\gamma\,\sigma^{-1}}\big(\sigma(z)\big)\,g_{\sigma}\big(z^{-1}\,x\big)
  \Big] \ ,
		\label{eq:transported-mu-M}
	\end{align}
	where in the last step we changed summation and integration variables $(y,\kappa)=(1,\gamma)\,{\textrm{\tiny$\bullet$}}\,(z,\sigma^{-1})$.
	This is exactly the product defined in \eqref{eq:prop:bimodule-from-Hopf-M}.
For the coproduct we compute
	\begin{align}
			\big(\Psi^{-1}\otimes\Psi^{-1}\big)\bigg[
          \Delta\bigg( \Psi\Big( \sum_{\varphi\in\Gamma}\, f_{\varphi}
          \Big) \bigg) \bigg] &=
			\big(\Psi^{-1}\otimes\Psi^{-1}\big)\Big( \big[
                                \big( (x,\gamma)\,,\,(y,\kappa)\big)
                                \longmapsto f_{\gamma\,\kappa}\big(
                                (\gamma\,\kappa)^{-1}\,(x\,\gamma(y))\big)
                                \big] \Big) \nonumber \\[4pt]
			&=\big[ \big(
                          (x,\gamma)\,,\,(y,\kappa)\big)\longmapsto
                          f_{\gamma\,\kappa}\big(\kappa^{-1}(x)\,y\big)
                          \big] \ ,
		\label{eq:transported-Delta-M}
	\end{align}
	which is exactly the coproduct described in \eqref{eq:prop:bimodule-from-Hopf-M}.
\end{proof}

\subsection{Gauging the \texorpdfstring{$\Out(G)$}{Out(G)}-symmetry}
\label{Sec:Orbifolding-Out-G}

We shall now compute the orbifold of two-dimensional Yang-Mills theory with gauge group $G$ using a defect corresponding to the $\Out(G)$-symmetry. We do this by introducing labels 
for defect junctions, and then compute the image of the projector which is the state space of the orbifold theory.
In order to verify that the orbifold theory is again a two-dimensional Yang-Mills theory 
--- with a different gauge group --- we compute the orbifold theory on the generators of
the category of bordisms with area and without defects $\Bord{\textrm{area}}$, which is given
by the commutative Frobenius algebra
structure on the state space.

The class $[\Phi]\in H^3(\Out(G),\Cb^{\times})$ obtained from the associator of a morphism category of the topological defect bicategory $\Bscr_{\textsf{\tiny YM}}$ is trivial, $[\Phi]=1$, so we can gauge any subgroup 
$\Gamma < \Out(G)$~\cite[Section~3]{FFSR:2009orb}, which we now fix.
In order to compute the orbifold theory of two-dimensional Yang-Mills theory with defect 
$M=\bigoplus_{\varphi\in\Gamma}\,L_{\varphi}$, we need to 
give labels for junctions of the defect labeled with $M$.
It is enough to give the labels for trivalent junctions because, owing
to the fact that $M$ is a strongly separable symmetric Frobenius algebra 
	in $\Bscr_{\textsf{\tiny YM}}(L,L)$, the value of the area-dependent quantum field theory
on a surface with a defect network is invariant under certain changes of
the defect network, which allow us to define junction fields with higher valency.

To a trivalent junction with an ingoing defect labeled with 
$L_{\varphi}$ and two outgoing defects labeled with $L_{\omega}$ and $L_{\gamma}$,
as in Figure~\ref{fig:trivalent-junctions}~$a)$,
we assign the family
\begin{align}
	\Big\lbrace\frac1{|\Gamma|}\,\delta_{\varphi,\gamma\,\omega}\,\eta_a^{L^2(G)}\Big\rbrace_{a
  \in \R_{>0}} \ \in \ \ 
	\ctimes_{L}\bar{L}_{\varphi}\otimes_{L} {L}_{\omega}
	\otimes_{L}L_{\varphi\,\omega}\simeq C\ell^2(G)\ ,
	\label{eq:trivalent-in-out-out}
\end{align}
where $a=a_1+a_2+a_3$ is the total area of the three individual surface
components.
\begin{figure}[htb]
\small
	\centering
	\input{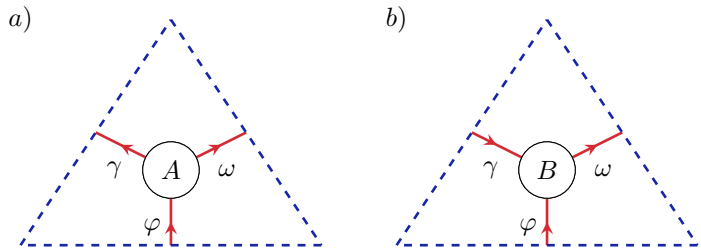}
	\caption{\small Triangles with trivalent junctions.}
	\label{fig:trivalent-junctions}
\normalsize
\end{figure}
Similarly, to a trivalent junction with ingoing defect lines labeled with $L_{\varphi}$
and $L_{\omega}$ and an outgoing defect line labeled with $L_{\gamma}$,
as in Figure~\ref{fig:trivalent-junctions}~$b)$,
we assign
the family 
\begin{align}
	\big\lbrace
  \delta_{\gamma\,\varphi,\omega}\,\eta_a^{L^2(G)}\big\rbrace_{a \in
  \R_{>0}} \ \in \ \ 
	\ctimes_{L}\bar{L}_{\varphi}\otimes_{L} \bar{L}_{\omega}
	\otimes_{L}L_{\varphi\,\omega}\simeq C\ell^2(G)\ .
	\label{eq:trivalent-in-in-out}
\end{align}
Here the chosen families correspond to the vertical identity morphisms
after horizontally composing the two equidirectional defects in the
bicategory of topological defects $\Bscr_{\textsf{\tiny YM}}$ from Section~\ref{Sec:Defect-Bicategory}.

From \cite[Section~5.3]{Runkel:2018aqft} we get 
\begin{lemma}
	The twisted sectors of the orbifold theory are given by
	\begin{align}
		\Hc=\bigoplus_{\varphi\in\Gamma}\,\Hc_{\varphi} \ ,
		\label{eq:lem:twisted-sectors}
	\end{align}
	where $\Hc_{\varphi}= \ \ctimes_{L}L_{\varphi}$
	is the space of square-integrable functions on $G$ which are invariant under
        twisted conjugation by $\varphi\in\Gamma$,
	see Section~\ref{Sec:Orbifold}.
	\label{lem:twisted-sectors}
\end{lemma}

We will need the morphisms
\begin{align}
	\input{iotaAB.tikz}
	\label{eq:iotaAB} 
\end{align}
\begin{lemma}
The state sum construction assigns to the parts of a cell decomposition
	\begin{center}
		\input{trivalent-junctions-cut.tikz}
	\end{center}
each with total
        area $a$, the respective morphisms
	\begin{align}
		\frac1{|\Gamma|}\,\delta_{\varphi,\gamma\,\omega}\,\iota_A\circ\eta_a^{L^2(G)}
		\qquad\text{and}\qquad
	\delta_{\gamma\,\varphi,\omega}\,\iota_B\circ\eta_a^{L^2(G)} \ .
		\label{eq:lem:trivalent-junctions-cut}
	\end{align}
	\label{lem:trivalent-junctions-cut}
\end{lemma}
\begin{proof}
	Let $\gamma:=\varphi\,\omega^{-1}$.
The morphism assigned to the part of the cell decomposition $a)$ is
\begin{align}
	\input{trivalent-junctions-cut-mor-a.tikz}
	\label{eq:trivalent-junctions-cut-mor-a}
\end{align}
Now we use the fact that $\iota_A$ 
cancels the idempotents $\Dc_0$ assigned to cylinders with parallel defect lines to obtain $\frac1{|\Gamma|}\,\iota_A\circ\eta^{L^2(G)}_a$.
One similarly computes the morphism associated to $b)$.
\end{proof}

Next we turn to the projector of \cite[Section~1]{Brunner:2013orb} whose image is 
the state space of the orbifold theory. 
The projector $\Pc$ is obtained by applying the defect area-dependent
quantum field theory on the 
cylinders in Figure~\ref{Fig:Defect-action} and taking the weak-coupling limit.
\begin{proposition}
	\begin{enumerate}
		\item 
	The projector $\Pc$ is given by
	\begin{align}
		\Pc=\frac1{|\Gamma|}\,\sum_{\omega\in\Gamma}\,\big(\omega^{-1}\big)^{*}:\Hc\longrightarrow
          \Hc\ ,
		\label{eq:prop:orbifold-projector}
	\end{align}
	which implements an action of $\Gamma$ on $\Hc$.
	\label{prop:orbifold-projector:1}
		\item
	The image of $\Pc$ is the subspace $\Hc^{\Gamma}$ of
        $\Gamma$-invariants under this action and there is an isomorphism
	\begin{align}
		\Hc^{\Gamma}\simeq C\ell^2(G\rtimes\Gamma)\ .
		\label{eq:prop:orbifold-projector-image}
	\end{align}
	\label{prop:orbifold-projector:2}
	\end{enumerate}
	\label{prop:orbifold-projector}
\end{proposition}
\begin{proof}
	For Part~\ref{prop:orbifold-projector:1}, take a cell
        decomposition of the cylinder in
        Figure~\ref{Fig:Defect-action} where we cut along the dashed lines: 
	\begin{center}
		\input{cylinder-decomp-2.tikz}
	\end{center}
Using Lemma~\ref{lem:trivalent-junctions-cut}, the morphism assigned
by the state sum area-dependent quantum field theory to the cylinder is
	\begingroup
	\allowdisplaybreaks
	\begin{align}
		&\input{orbifold-projector-simplified-1.tikz}\nonumber\\[4pt]
		& \hspace{1cm} \input{orbifold-projector-simplified-2.tikz}
		\label{eq:orbifold-projector-simplified}
	\end{align}
	\endgroup
	Summing over $\varphi,\omega\in\Gamma$ we get \eqref{eq:prop:orbifold-projector}.

	For Part~\ref{prop:orbifold-projector:2}, let $h\in\Hc^{\Gamma}$ and write $h=\sum_{\varphi\in\Gamma}\,h_{\varphi}$
	for its components. 
	We define the map
	\begin{align}
		\Psi:\Hc^{\Gamma}\longrightarrow C\ell^2(G\rtimes
                                   \Gamma) \ , \quad
		h\longmapsto \big[ (x,\alpha)\mapsto h_{\alpha}(x) \big] \ .
		\label{eq:H-Gamma-Cl-iso}
	\end{align}
We show that the image of the map $\Psi$ indeed lands in $C\ell^2(G\rtimes \Gamma)$.
	Since $h$ is invariant under the action of $\Gamma$ we have
	\begin{align}
	\omega\cdot h=\sum_{\varphi\in\Gamma}\,\omega\cdot h_{\varphi}
		=\sum_{\gamma\in\Gamma}\,h_{\gamma}=h \ .
		\label{eq:hinvariant}
	\end{align}
	Because $\omega\cdot h_{\varphi}\in\Hc_{\omega\,\varphi\,\omega^{-1}}$,
	we get 
	\begin{align}
		h_{\varphi}\circ
          \omega^{-1}=\omega\cdot h_{\varphi}=h_{\omega\,\varphi\,\omega^{-1}}
          \ .
		\label{eq:hinvariant-comp}
	\end{align}
	Let $(x,\varphi),(y,\omega)\in G\rtimes\Gamma$ and $f:=\Psi(h)$.
	Then 
	\begin{align}
		f\big( (x,\varphi)\,{\textrm{\tiny$\bullet$}}\,(y,\omega)\,{\textrm{\tiny$\bullet$}}\,(x,\varphi)^{-1} \big)&=
		f\big( x\,\varphi(y)\,(\varphi\,\omega\,\varphi^{-1})(x^{-1})\,,\,\varphi\,\omega\,\varphi^{-1} \big)
		\nonumber \\[4pt]
&:=
		h_{\varphi\,\omega\,\varphi^{-1}}\big(
                x\,\varphi(y)\,(\varphi\,\omega\,\varphi^{-1})(x^{-1})
                \big) \nonumber \\[4pt]
		&=
                h_{\varphi\,\omega\,\varphi^{-1}}\big(\varphi(y)\big)
                \nonumber \\[4pt]
&=
		h_{\omega}(y) \nonumber \\[4pt]
		&=:
		f(y,\omega) \ ,
		\label{eq:Psi-lands-in-Cl}
	\end{align}
	where in the third equality we used the twisted conjugation property of elements in
	$\Hc_{\varphi\,\omega\,\varphi^{-1}}$, and in the fourth
        equality we used \eqref{eq:hinvariant-comp}.

	Now we define 
	\begin{align}
		\Phi:C\ell^2(G\rtimes\Gamma)&\longrightarrow\Hc^{\Gamma}
                                              \ , \quad
		f\longmapsto h^f=\sum_{\varphi\in\Gamma}\,h^f_{\varphi}\ ,
		\label{eq:Cl-H-Gamma-iso}
	\end{align}
	where $h^f_{\varphi}(x)=f(x,\varphi)$.
We show that the image of the map $\Phi$ indeed lands in $\Hc^{\Gamma}$.
	First we show that $h^f_{\varphi}\in\Hc_{\varphi}$ for $\varphi\in\Gamma$. 
	Let $x,y\in G$ and compute
	\begin{align}
		h^f_{\varphi}\big( x\,y\,\varphi(x^{-1}) \big)
		:=
		f\big( x\,y\,\varphi(x^{-1}),\varphi \big) 
		=f\big( (x,\id_G)\,{\textrm{\tiny$\bullet$}}\,(y,\varphi)\,{\textrm{\tiny$\bullet$}}\,(x,\id_G)^{-1}
                  \big) 
		=
		f(y,\varphi) 
=h^f_{\varphi}(y) \ ,
		\label{eq:Phi-lands-in-H-Gamma-1}
	\end{align}
	where in the third equality we used the twisted conjugation invariance of $f$.
	Next we show that $h^f$ is $\Gamma$-invariant:
	\begin{align}
		(\omega\cdot h^f)(x)&=\sum_{\varphi\in\Gamma}\,(\omega\cdot
                h^f_{\varphi})(x) \nonumber \\[4pt]
&=
		\sum_{\varphi\in\Gamma}\,
                h^f_{\varphi}\big(\omega^{-1}(x)\big) \nonumber \\[4pt]
		&=\sum_{\varphi\in\Gamma}\, f\big(
                \omega^{-1}(x),\varphi \big) \nonumber \\[4pt]
		&=
		\sum_{\varphi\in\Gamma}\, f\big(
                (1,\omega)\,{\textrm{\tiny$\bullet$}}\,(\omega^{-1}(x),\varphi)\,{\textrm{\tiny$\bullet$}}\,(1,\omega)^{-1}\big)
                \nonumber \\[4pt]
		&=\sum_{\varphi\in\Gamma}\,
                f\big(x,\omega\,\varphi\,\omega^{-1}\big) \nonumber \\[4pt]
		&=\sum_{\varphi'\in\Gamma}\, f\big(x,\varphi'\big)
                \nonumber \\[4pt]
&=h^f(x) \ ,
		\label{eq:Phi-lands-in-H-Gamma-2}
	\end{align}
where in the fourth equality we used again the twisted conjugation
invariance of $f$ and changed summation variable in the sixth equality.
	
	Clearly the maps $\Psi$ and $\Phi$ are inverse to each other.
\end{proof}

\begin{theorem}
	The orbifold theory of two-dimensional Yang-Mills theory with
        gauge group $G$ and with orbifold defect
\begin{align}
\bigoplus_{\varphi\in\Gamma}\,L_{\varphi}
\end{align}
is two-dimensional
        Yang-Mills theory with gauge group $G\rtimes\Gamma$.
	\label{thm:orbifold-2dYM}
\end{theorem}
\begin{proof}
	We need to compute the regularised Frobenius algebra structure on $C\ell^2(G\rtimes \Out(G))$
	given by the orbifold theory. This is done by computing the 
	orbifold theory on the generators of $\Bord{\mathrm{area}}$.
	The computations are similar to those for the cylinder in the proof of
	Proposition~\ref{prop:orbifold-projector}, so here we provide less details.

	For the cup with area $a$ we pick the defect network 
	and cell decomposition
	\begin{center}
		\input{cup-decomp-2.tikz}
	\end{center}
where we identify the two dashed edges on the two sides. The value of
the state sum area-dependent quantum field theory on this is
	\begin{align}
		\input{orbifold-unit-2.tikz}
		\label{eq:orbifold-unit}
	\end{align}
	We similarly obtain
\begin{align}
\input{orbifold-counit.tikz}
\end{align}
for the value of the area-dependent quantum field theory on the cap with the defect network and cell decomposition
	\begin{center}
		\input{cap-decomp-2.tikz}
	\end{center}

	Finally let us turn to the pair of pants with two ingoing
        circles and
        one outgoing circle:
	\begin{center}
		\input{pants-decomp.tikz}
	\end{center}
	To this decomposition the state sum area-dependent quantum
        field theory assigns the morphism 
	\begingroup
	\allowdisplaybreaks
	\begin{align}
	&\input{orbifold-pants-1.tikz}\nonumber\\[4pt]
		& \hspace{3cm} \input{orbifold-pants-2.tikz}
		\label{eq:orbifold-pants}
	\end{align}
	\endgroup
	Let $f,g\in C\ell^2(G\rtimes\Gamma)\simeq \Hc^{\Gamma}$ with components
	$f_{\omega}=f(\,\cdot\,,\omega)$ and
        $g_\omega=g(\,\cdot\,,\omega)$ for $\omega\in\Gamma$. The
        morphism in \eqref{eq:orbifold-pants} acts on these functions as
	\begin{align}
		f\otimes g\longmapsto
		\frac1{|\Gamma|^2}\,\sum_{\varphi,\omega,\varrho,\nu\in\Gamma}\,
		\big(f_{\nu\,\omega\, \nu^{-1}}\circ(\varrho\,\varphi\, \varrho^{-1})\big)*g_{\varrho\,\varphi\, \varrho^{-1}}
		=\sum_{\sigma,\kappa\in\Gamma}\,
		(f_{\kappa}\circ \sigma)*g_{\sigma}
		=\sum_{\sigma,\kappa\in\Gamma}\,
		f_{\sigma^{-1}\,\kappa\,\sigma}*g_{\sigma}
		\label{eq:orbifold-pants-value}
	\end{align}
	which maps 
	\begin{align}
		(x,\gamma)\longmapsto
		\sum_{\sigma,\kappa\in\Gamma}\, \big(f_{\sigma^{-1}\,\kappa\,\sigma}*g_{\sigma}\big) (x)\,\delta_{\kappa\,\sigma,\gamma}
		=\sum_{\sigma\in\Gamma}\, \big(f_{\sigma^{-1} \,
          \gamma}*g_{\sigma}\big)(x) \ ,
		\label{eq:orbifold-pants-value2}
	\end{align}
	as $f_{\sigma^{-1}\,\kappa\,\sigma}*g_{\sigma}\in\Hc_{\kappa\,\sigma}$.
	On the other hand we have
	\begin{align}
		(f*g) (x,\gamma)&=
		\frac1{|\Gamma|}\,\sum_{\kappa\in\Gamma}\,\int_{G}\,
                \dd y \ 
		f(y,\kappa)\,g\big((y,\kappa)^{-1}\,{\textrm{\tiny$\bullet$}}\,(x,\gamma)\big)
                \nonumber \\[4pt]
                &=\frac1{|\Gamma|}\,\sum_{\kappa\in\Gamma}\,\int_{G}\,
                \dd y \ 
		f_{\kappa}(y)\,g_{\kappa^{-1}\,\gamma}\big(\kappa^{-1}(y^{-1}\,x)\big)
                \nonumber \\[4pt]
		&=\frac1{|\Gamma|}\,\sum_{\kappa\in\Gamma}\,\int_{G}\,
                \dd y \ 
		f_{\kappa}(y)\,g_{\gamma\,\kappa^{-1}}\big(y^{-1}\,x\big)
                \nonumber \\[4pt]
		&= \frac1{|\Gamma|}\,\sum_{\sigma\in\Gamma}\,
		\big(f_{\sigma^{-1} \, \gamma}*g_{\sigma}\big)(x) \ .
		\label{eq:convprod}
	\end{align}
	Altogether we have shown that multiplying by $|\Gamma|$ is an
        isomorphism of regularised Frobenius algebras
	$C\ell^2(G\rtimes\Gamma)\xrightarrow{ \ \simeq \ }\Hc^{\Gamma}$.
\end{proof}

\subsection{Orbifold equivalence and the backwards orbifold}
\label{Sec:Backwards-Orbifold}

We can translate the statement of Theorem~\ref{thm:orbifold-2dYM} into an adjoint equivalence in the
orbifold completion $\Bscr_{\textsf{\tiny
    YM}}^\mathrm{orb}$~\cite{Carqueville:2012orb} of the topological
defect bicategory of two-dimensional Yang-Mills theories
$\Bscr_{\textsf{\tiny YM}}$, which is idempotent complete by
Proposition~\ref{prop:B-idempotent-complete}. Since the objects of
$\Bscr_{\textsf{\tiny YM}}$ are bimodules in the symmetric monoidal
category $\Hilb$, and the left and right duality morphisms can be
related by the symmetric braiding, it follows that the defect
bicategory $\Bscr_{\textsf{\tiny YM}}$ is pivotal, similarly to the
case of topological field theories.
\begin{proposition}
There are adjoint equivalences in $\Bscr_{\textsf{\tiny YM}}^\mathrm{orb}$:
	\begin{align}
		{}_{K}K_{L}: \left( L,{}_{L}K_{L} \right)
          \rightleftarrows \left( K,{}_{K}K_{K} \right) :{}_{L}K_{K}
          \qquad \mbox{and} \qquad 	{}_{K}K_{L}:
		\left( L,{}_{L}L_{L} \right)
		\rightleftarrows
		\left( K, {}_{K}K\otimes_{L} K_{K} \right)
		:{}_{L}K_{K}\ .
	\end{align}
\label{prop:adjoint-equiv}
\end{proposition}
\begin{proof}
	Observe that ${}_{L}K_{L}\simeq {}_{L}K\otimes_{K}K_{L}$.
	Then use \cite[Proposition~4.4]{Carqueville:2012orb} get the second adjoint equivalence.
\end{proof}

By \cite[Proposition~4.3]{Carqueville:2012orb}, $K\otimes_{L^2(G)} K$ has 
the structure of a strongly separable symmetric Frobenius algebra, 
which we will describe in more detail now.
We will see that $K\otimes_{L^2(G)} K$ is a Wilson line defect,
which is isomorphic to a defect coming from a group symmetry exactly when $\Gamma$ is abelian.
\begin{proposition}
	\begin{enumerate}
	\item $K\otimes_L K\simeq K\otimes H$ as $K$--$K$-bimodules.
			The left $K$-action is by multiplication on the $K$ factor and the
			right $K$-action is
			\begin{align}
				\input{KH-right-action.tikz}
				\label{eq:KH-right-action}
			\end{align}
			where we used the isomorphism 
			$\Phi:L\rtimes H \xrightarrow{~\simeq~} K$
                        from Part~\ref{prop:bimodule-from-Hopf:1}
			of Proposition~\ref{prop:bimodule-from-Hopf}.
			Denote this $K$--$K$-bimodule by $(K\otimes H)_\mathrm{ind}$.
			\label{prop:reverse-orbifold-bimod:1}
		\item Consider $H=L^2(\Gamma)$ as a left $\Gamma$-module with action given by
			$(\gamma\cdot\phi)_\omega=\phi_{\omega\,\gamma}$ for 
			$\omega,\gamma\in\Gamma$ and $\phi\in H$, and 
			consider the pullback of $H$ along the projection 
			$G\rtimes\Gamma\longrightarrow\Gamma$ which we denote again by $H$.
			Write $(K\otimes H)_\mathrm{Wilson}$ for the $K$--$K$-bimodule 
			structure given by Example~\ref{ex:bimodule-Wilson}.
			The left $K$-action is by multiplication on the 
			$K$ factor and the right $K$-action is given by
			\begin{align}
				\input{KH-Wilson-action.tikz}
				\label{eq:KH-Wilson-action}
			\end{align}
			\label{prop:reverse-orbifold-bimod:2}
		\item The map 
			\begin{align}
				\input{KH-iso-Ind-Wil.tikz}
				\label{eq:KH-iso-Ind-Wil}
			\end{align}
			is a $K$--$K$-bimodule isomorphism $\Psi:(K\otimes H)_\mathrm{ind}\longrightarrow(K\otimes H)_\mathrm{Wilson}$.
			\label{prop:reverse-orbifold-bimod:3}
		\item The bimodule $(K\otimes H)_\mathrm{Wilson}$ is a direct sum of invertible
			bimodules if and only if $\Gamma$ is abelian.
			\label{prop:reverse-orbifold-bimod:4}
	\end{enumerate}
	\label{prop:reverse-orbifold-bimod}
\end{proposition}
\begin{proof}
	For Part~\ref{prop:reverse-orbifold-bimod:1},
	the idempotent $\Dc_0^{K,K}$ projecting onto the relative tensor product is the weak-coupling limit
	\begin{align}
		\input{idempotKLK.tikz}
		\label{eq:idempotKLK}
	\end{align}
	We can factorise this as $\Dc_0^{K,K}=\iota\circ\pi$, where
	\begin{align}
		\input{iota-pi-KLK.tikz}
		\label{eq:iota-pi-KLK}
	\end{align}
	We first show $\pi\circ\iota=\id_{K\otimes H}$:
	\begin{align}
		\input{pi-iota-id-KLK-easy.tikz} \ .
		\label{eq:pi-iota-id-KLK}
	\end{align}
	Now we show $\iota\circ\pi=\Dc_0^{K,K}$:
	\begin{align}
		\input{iota-pi-D0-KLK.tikz} \ .
		\label{eq:iota-pi-D0-KLK}
	\end{align}
	The induced action can be computed in a similar way as
	\begin{align}
		\input{induced-KK-action.tikz}
		\label{eq:induced-KK-action}
	\end{align}
	and after writing out the right action of $L$ on $K$ we get \eqref{eq:KH-right-action}.

	For Part~\ref{prop:reverse-orbifold-bimod:2}
	we compute the right $K$-action on $(K\otimes H)_\mathrm{Wilson}$
	for $f,g\in L$ and
        $\delta_{\varphi},\delta_{\omega},\delta_{\gamma}\in H$ with $\varphi,\omega,\gamma\in\Gamma$:
	\begin{align}
		& (f\otimes\delta_{\varphi}\otimes\delta_{\omega})\cdot(g\otimes\delta_{\gamma})
          \\ & \hspace{3cm} =
		\left[ (x,\sigma)\longmapsto \int_{G}\, \dd y \
                                                                                                \frac1{|\Gamma|}\, \sum_{\kappa\in\Gamma}\,
		f\big(x\,(\sigma\,\kappa^{-1})(y^{-1})\big)\,\delta_{\varphi}\big(\sigma\,\kappa^{-1}\big)\,
		\big(\kappa^{-1}\cdot\delta_{\omega})\, g(y)\,\delta_{\gamma}(\kappa)\right]\\[4pt]
		& \hspace{3cm} =\left[ (x,\sigma)\longmapsto \int_{G}\, \dd y \ \frac1{|\Gamma|}\,\sum_{\kappa\in\Gamma}\,
		f\big(x\,\varphi(y^{-1})\big)\,\delta_{\varphi\,\gamma}(\sigma)\,
		\big(\gamma^{-1}\cdot\delta_{\omega}\big)\, g(y)\,\delta_{\gamma}(\kappa)\right]\\[4pt]
		& \hspace{3cm} =\left[ (x,\sigma)\longmapsto
                  \frac1{|\Gamma|}\,\int_{G}\,\dd y \ 
		f\big(x\,\varphi(y^{-1})\big)\delta_{\varphi\,\gamma}(\sigma)\,
		\delta_{\omega\,\gamma}\, g(y)\right]\\[4pt]
		& \hspace{3cm} =\left[ (x,\sigma)\longmapsto
                  \frac1{|\Gamma|}\,\int_{G}\, \dd z \ 
		f\big(x\,z^{-1}\big)\,\delta_{\varphi\,\gamma}(\sigma)\,
		\delta_{\omega\,\gamma}\, g\big(\varphi^{-1}(z)\big)\right]\\[4pt]
		& \hspace{3cm} =\big(f*(\delta_{\varphi}\cdot
                  g)\big)\otimes\delta_{\varphi\,\gamma}\otimes\delta_{\omega\,\gamma}
                  \ ,
	\label{eq:KH-Wilson-action-cmputation}
	\end{align}
	where in the first step we used the delta-functions, in the
        second step we used the action of $\Gamma$ on $H$, and finally
        the invariance of the integral.
	This is exactly the right $K$-action in \eqref{eq:KH-Wilson-action}.

	For Part~\ref{prop:reverse-orbifold-bimod:3}, we use
        Part~\ref{prop:reverse-orbifold-bimod:1} to read off the right $K$-action
        on $(K\otimes H)_\mathrm{ind}$ to be
	\begin{align}
		(f\otimes\delta_{\varphi}\otimes\delta_{\omega})\cdot(g\otimes\delta_{\gamma})
		=\big(f*(\delta_{\varphi\,\omega}\cdot
          g)\big)\otimes\delta_{\varphi}\otimes\delta_{\omega\,\gamma}
          \ .
		\label{eq:KH-induced-action-computation}
	\end{align}
	Since $\Psi$ obviously commutes with the left $K$-actions, we only need to show that $\Psi$ commutes with the right $K$-actions:
	\begin{align}
		\Psi(f\otimes\delta_{\varphi}\otimes\delta_{\omega})
                \cdot (g\otimes\delta_{\gamma})
		&=( f\otimes\delta_{\varphi\,\omega}\otimes\delta_{\omega})\cdot(g\otimes\delta_{\gamma})
		=\big(f*(\delta_{\varphi\,\omega}\cdot
                g)\big)\otimes\delta_{\varphi\,\omega\,\gamma}\otimes\delta_{\omega\,\gamma}
                \ , \\[4pt]
		\Psi\big( (f\otimes\delta_{\varphi}\otimes\delta_{\omega})\cdot(g\otimes\delta_{\gamma}) \big)
		&=
		\Psi\big( (f*(\delta_{\varphi\,\omega}\cdot g))\otimes\delta_{\varphi}\otimes\delta_{\omega\,\gamma} \big)
		=\big(f*(\delta_{\varphi\,\omega}\cdot
                g)\big)\otimes\delta_{\varphi\,\omega\,\gamma}\otimes\delta_{\omega\,\gamma}
                \ .
		\label{eq:KH-intertwiner-computation}
	\end{align}
	Clearly $\Psi$ is an isomorphism.

	Part~\ref{prop:reverse-orbifold-bimod:4} follows from the facts that 
	 $L^2(\Gamma)\simeq
         \bigoplus_{c\in\widehat{\Gamma}}\,V_c^{\oplus v_c}$ as $\Gamma$-modules, where
         $v_c=\dim V_c$ and $V_c$ is a representative of the conjugacy
         class $c$, and that
        all simple $\Gamma$-modules are one-dimensional if and only if $\Gamma$ is abelian.
\end{proof}

As a consequence of Propositions~\ref{prop:adjoint-equiv} and~\ref{prop:reverse-orbifold-bimod} we get the defect for the backwards orbifold:
\begin{theorem}
	The orbifold theory of two-dimensional Yang-Mills theory with
        gauge group $G\rtimes\Gamma$ and with orbifold
	Wilson line defect 
	\begin{align}
		L^2(G\rtimes \Gamma)\otimes L^2(\Gamma)
		\label{eq:thm:backwards-orbifold-2dYM}
	\end{align}
	is two-dimensional Yang-Mills theory with gauge group $G$.
	\label{thm:backwards-orbifold-2dYM}
\end{theorem}

\newpage

\phantomsection
\addcontentsline{toc}{section}{References}

\bibliographystyle{halphamod}

\end{document}